\documentclass[a4paper]{amsart}

\usepackage[latin1]{inputenc}  
\usepackage[T1]{fontenc}       

\usepackage{amsmath}
\usepackage{amsthm}
\usepackage{amssymb}
\usepackage{amscd}
\usepackage{amsfonts}
\usepackage{textcomp}
\usepackage{tikz, subfigure, xcolor} 
\usepackage{dsfont}

\usepackage{dsfont}

\newtheorem{theorem}{Theorem}[section]
\newtheorem{lemma}[theorem]{Lemma}
\newtheorem{proposition}[theorem]{Proposition}
\newtheorem{corollary}[theorem]{Corollary}
\newtheorem{remark}[theorem]{Remark}

\newtheorem{example}[theorem]{Example}
\newtheorem{definition}[theorem]{Definition}

\providecommand{\abs}[1]{\lvert#1\rvert}

\providecommand{\norm}[1]{\lVert#1\rVert}

\providecommand{\Id}[1]{\mbox{Id}}
\providecommand{\deg}[1]{\mbox{deg}}
\providecommand{\diag}[1]{\mbox{diag}}
\providecommand{\Ran}{\mbox{Ran}}

\providecommand{\det}[1]{\mbox{det}}
\providecommand{\R}{\mathds{R}}
\providecommand{\C}{\mathds{C}}
\providecommand{\Z}[1]{\mathds{Z}}
\providecommand{\K}[1]{\mathds{K}}
\DeclareMathOperator{\Ie}{\mathcal{I}} 
\DeclareMathOperator{\Ee}{\mathcal{E}} 
\DeclareMathOperator{\Qe}{\mathcal{Q}} 
\DeclareMathOperator{\We}{\mathcal{W}} 
\DeclareMathOperator{\He}{\mathcal{H}} 
\DeclareMathOperator{\Ke}{\mathcal{K}}
\DeclareMathOperator{\De}{\mathcal{D}} 
\DeclareMathOperator{\Ge}{\mathcal{G}} 
\DeclareMathOperator{\Pe}{\mathcal{P}} 
\DeclareMathOperator{\Me}{\mathcal{M}} 
\DeclareMathOperator{\Be}{\mathcal{B}} 
\DeclareMathOperator{\cL}{\mathcal{L}} 
\DeclareMathOperator{\N}{\mathds{N}} 
\DeclareMathOperator{\grad}{grad}

\DeclareMathOperator{\Ker}{Ker}

\renewcommand{\Im}{{\ensuremath{\mathrm{Im\,}}}} 
\renewcommand{\Re}{{\ensuremath{\mathrm{Re\,}}}} 
\renewcommand{\div}{\mathrm{div}\,}

\DeclareMathOperator{\Cm}{\mathfrak{C}}
\DeclareMathOperator{\Jm}{\mathfrak{J}} 
\DeclareMathOperator{\jm}{\mathfrak{j}} 

\DeclareMathOperator{\au}{\underline{a}}
\DeclareMathOperator{\bu}{\underline{b}}
\DeclareMathOperator{\sgn}{\mbox{sgn}}
\DeclareMathOperator{\Dom}{\mbox{Dom}}

\author{Amru Hussein}
\address{A.~Hussein, FB 08 - Institut f\"{u}r Mathematik,
Johannes Gutenberg-Universit\"{a}t Mainz,
Staudinger Weg 9,
D-55099 Mainz,
Germany}
\email{hussein@mathematik.uni-mainz.de}

\title[Sign--indefinite operators on metric graphs]{Sign--indefinite second order differential operators on finite metric graphs.}

\begin{document}

\begin{abstract}
The question of self--adjoint realizations of sign--indefinite second order differential operators is discussed in terms of a model problem. Operators of the type $-\frac{d}{dx} \sgn  (x) \frac{d}{dx}$ are generalized to finite, not necessarily compact, metric graphs. All self--adjoint realizations are parametrized using methods from extension theory. The spectral and scattering theory of the self--adjoint realizations are studied in detail. 
\end{abstract}
  \maketitle
\section{Introduction}

Differential expressions of the form 
\begin{eqnarray}\label{eq1IndefQG}
 -\div  A(\cdot) \grad
& \mbox{with}  & A(x)=\begin{cases} +1, & x\in \Omega_+,\\ 
											 -1, & x\in \Omega_-,
                    \end{cases} 
\end{eqnarray} 
for domains $\Omega=\Omega_+ \cup \Omega_-\subset \R^d$ appear in different contexts. The Poisson problem related to this operator is of physical relevance. It appears in the mathematical description of light propagation through regions with thresholds of  materials with negative and  positive refraction indices, see \cite{Bouchitte} and the references therein. The quasi--static limit of the Maxwell equation yields the mentioned Poisson problem. Another field of application is solid-state physics. In the effective mass approximation the effective mass tensor can be negative as well, see for example \cite{Allaire2005} and the references therein. This yields indefinite differential operators, which involve expressions of the form \eqref{eq1IndefQG} and in this context the Schr\"{o}dinger equation is considered.

An approach that has turned out to be very fruitful for the study of differential equations which involve the expression $-\div  A(\cdot) \grad$ with strictly positive coefficient matrices $A(\cdot)$, is to consider self--adjoint operators associated with these differential expression and their spectral resolution. The question is, if it is possible to find self--adjoint realizations also for sign--indefinite coefficients $A(\cdot)$. Unbounded operators with sign--changing coefficients in the highest order terms are not well studied. To make a starting point a new model problem is introduced: Laplace operators on finite metric graphs with changing sign. The aim of this work is to present generalizations of the operator
\begin{eqnarray}\label{eq2IndefQG}
-\frac{d}{dx} \sgn  (x) \frac{d}{dx} & \mbox{with}& \sgn (x)=\begin{cases} +1, & x \geq 0,\\ 
							-1, & x<0	
                    \end{cases}
\end{eqnarray}
to finite metric graphs. This initial problem can be tackled using extension theory. This approach allows straight forward generalizations from intervals to finite metric graphs. It yields a formal operator which admits the block operator matrix representation
\begin{eqnarray*}
\tau=\begin{bmatrix}
-\Delta & 0 \\ 0 & +\Delta
\end{bmatrix},
\end{eqnarray*}
where $-\Delta$ denotes the one dimensional Laplace operator, and all self--adjoint realisations are completely characterized. On finite metric graphs this gives a model, which on the one hand is still explicitly solvable, but on the other hand allows to describe more complicated geometries than the one dimensional case does, which deals only with intervals. A similar procedure has been performed already very successfully for the sign--definite one dimensional Laplacian. This has become famous under the name ``quantum graphs''. For further references on this topic see for example \cite{QG} and the references given therein. It turns out that there is a one--to--one correspondence between the self--adjoint realizations of the Laplacian on graphs and the self--adjoint realizations for the model--operator discussed here. 

The aim is to construct self--adjoint realizations of expressions of the form \eqref{eq1IndefQG} and \eqref{eq2IndefQG} in terms of extension theory and to develop a better understanding of the spectral and scattering properties of sign--indefinite differential operators. Some aspects of the problem are included in the general setting of Sturm-Liouville theory, but they have been considered to have no practical relevance, as remarked in \cite[Section 17.E]{WeidmannODE}. 

The motivation to consider expressions of the form \eqref{eq1IndefQG} in the present paper is caused mainly by the fast development in the field of so called metamaterials -- materials with negative optical refraction index. This has been initiated by  V.~G.~Veselago, see \cite{Veselago1968}, who was the first to consider -- hypothetically -- the effect of a negative refraction index. He had predicted that this feature would yield new effects including for example backward waves, an inverse law of refraction and an inverse Doppler effect. Today metamaterials can be constructed as artificial materials, whose optical properties are made useful for applications. An overview of recent developments is given in \cite{Pendry}, recall also the article \cite{Bouchitte} and the references given therein. 

Note that the model problem \eqref{eq2IndefQG} is part of a range of problems related to the ordinary differential expressions 
\begin{eqnarray*}
-\frac{d^2}{dx^2}, &   & -\frac{d}{dx} \sgn  (x) \frac{d}{dx}, \\ -\sgn (x) \frac{d^2}{dx^2} & \mbox{and} & -\sgn (x) \frac{d}{dx} \sgn  (x) \frac{d}{dx}.
\end{eqnarray*}
Each is exemplary for certain difficulties and methods used. Usually the operators written in the first line are considered in Hilbert spaces, whereas the operators in the second line are considered in Krein spaces, see for example \cite{Najman1995,Zettl}. The operator $-\sgn (x) \tfrac{d^2}{dx^2}$ has applications in the effective mass approximation too, compare \cite{Znojil}. 

This note is organized as follows: first the basic concepts are introduced and the notation is fixed. The subsequent section is devoted to the characterization of all self--adjoint boundary conditions. In short the corresponding question for Krein--space--self--adjointness is discussed as well. Section~\ref{secform} puts the question of self-adjoint extensions into the context of indefinite quadratic forms, and in Section~\ref{secext} the operators discussed are put into the general context of extension theory. Section \ref{secres} gives explicit formulae for eigenvalues, resonances and resolvents. In Section \ref{secscat} the scattering properties of the system are discussed. The wave operators are computed as well as the scattering matrix. The scattering matrix can be computed in certain cases in terms of a generalized star product.

The content of this note has developed in parallel to the work on the article \cite{VK12} which is in preparation, and it is part of the author's PhD thesis, see \cite[Chapter 4]{Ich}. Indefinite operators on graphs provide an illustration and can serve as a source of examples for indefinite operators of the form \eqref{eq1IndefQG} on manifolds and intervals involving more general coefficients $A(\cdot)$. 

\subsection*{Acknowledgements}
I would like to thank David Krej\v{c}i\v{r}\'{i}k for communicating the problem of sign--indefinite differential operators, its physical background, helpful discussions as well as remarks and especially for kind hospitality during my visit to the Doppler Institute in Prague in February 2012. I am grateful as well to Vadim Kostrykin and to Stephan Schmitz from whom I benefited in many ways during the still ongoing work on the article \cite{VK12}.

\section{Basic concepts}\label{sec2}

Instead of considering the operator given in \eqref{eq2IndefQG} in the Hilbert space $L^2(\R, dx)$, one considers in the equivalent space 
$$L^2([0,\infty), dx_1) \oplus L^2([0,\infty), dx_2) \quad (\cong L^2(\R, dx))$$ 
the formal differential operator 
$$\tau u= (+u_1^{\prime \prime}, -u_2^{\prime \prime}).$$ 
This allows to define a closed symmetric operator with equal deficiency indices $(2,2)$. Hence there exist self-adjoint extensions. The operator $\tau$ is generalized from intervals to finite metric graphs which delivers indefinite symmetric second order differential operators with finite deficiency indices. In the following the basic concepts are presented. The notation is mainly borrowed from the works of V.~Kostrykin and R.~Schrader on Laplacians on metric graphs, see for example \cite{VKRS1999}.

\subsection{Finite metric graphs} Here and in the following a graph is a $4$-tuple $$\Ge = \left( V, \Ie,\Ee, \partial \right),$$ where $V$ denotes the set of \textit{vertices}, $\Ie$ the set of \textit{internal edges} and $\Ee$ the set of \textit{external edges}, where the set $\Ee \cup \Ie$ is summed up in the notion \textit{edges}. The \textit{boundary map} $\partial$ assigns to each internal edge $i\in \Ie$ an ordered pair of vertices $\partial (i)= (v_1,v_2)\in V\times V$, where $v_1$ is called its \textit{initial vertex} and $v_2$ its \textit{terminal vertex}. Each external edge $e\in \Ee$ is mapped by $\partial$ onto a single, its initial, vertex. A graph is called finite if $\abs{V}+\abs{\Ie}+\abs{\Ee}<\infty$ and a finite graph is called \textit{compact} if $\Ee=\emptyset$ holds.

The graph $\Ge$ is endowed with the following metric structure. Each internal edge $i\in \Ie$ is associated with an interval $[0,a_i]$ with $a_i>0$, such that its initial vertex corresponds to $0$ and its terminal vertex to $a_i$. Each external edge $e\in \Ee$ is associated to the half line $[0,\infty)$, such that $\partial(e)$ corresponds to $0$. The numbers $a_i$ are called \textit{lengths} of the internal edges $i\in \Ie$ and they are summed up into the vector $\au=\{a_i\}_{i\in \Ie}\in \R_+^{\abs{\Ie}}$. The 2-tuple  $(\Ge,\au)$ consisting of a finite graph endowed with a metric structure defines a metric space called \textit{metric graph}. The metric on $(\Ge,\au)$ is defined via minimal path lengths on it. 

One distinguishes two types of edges
\begin{eqnarray*}\index{$\Ee_{\pm}$}\index{$\Ie_{\pm}$}
\mathcal{E}=\mathcal{E}_+ \dot\cup \mathcal{E}_- & \mbox{and} & \mathcal{I}=\mathcal{I}_+ \dot\cup \mathcal{I}_- ,
\end{eqnarray*}
where $\Ee_+$ ($\Ee_-$) and $\Ie_+$ ($\Ie_-$)  are called the \textit{positive (negative) external edges} and \textit{positive (negative) internal edges}, respectively. The set $\Ee_+ \cup \Ie_+$ ($\Ee_- \cup \Ie_-$) denotes the \textit{positive (negative) edges}. This partition defines two sub--graphs 
\begin{eqnarray*}\index{$\Ge_{\pm}$}
\Ge_-, \Ge_+ \subset \Ge, &\mbox{where}& \Ge_{\pm}=(V_{\pm} ,\Ie_{\pm} ,\Ee_{\pm} ,\partial\mid_{(\Ee_{\pm}\cup \Ie_{\pm})}) 
\end{eqnarray*}\index{$\au_{\pm}$}
with $V_{\pm}=\partial (\Ee_{\pm}\cup \Ie_{\pm})$. Let be $\au_{\pm}=\{a_i\}_{i\in \Ie_{\pm}}$, then $(\Ge_{\pm},\au_{\pm})$ are metric graphs. For brevity set 
\begin{eqnarray*}
n=2\abs{\Ie_+} + \abs{\Ee_+}&  \mbox{and} & m=2\abs{\Ie_-} + \abs{\Ee_-}.
\end{eqnarray*}
\subsection{Function spaces} 
Given a finite metric graph $(\Ge,\au)$ one considers the Hilbert space
\begin{eqnarray*}
\He \equiv \He(\Ee,\Ie,\au)= \He_{\Ee} \oplus \He_{\Ie}, & \He_{\Ee}= \displaystyle{\bigoplus_{e\in\Ee} \He_e,} & \He_{\Ie}= \bigoplus_{i\in\Ie} \He_i,
\end{eqnarray*}     
where $\He_j= L^2(I_j)$ with 
$$I_j= \begin{cases} [0,a_j], & \mbox{if} \ j\in \Ie, \\ [0,\infty), &\mbox{if} \ j\in \Ee.  \end{cases}$$
Again one has a decomposition into two subspaces
$$\He(\Ee,\Ie,\au) = \He(\Ee_+,\Ie_+,\au_+) \oplus \He(\Ee_-,\Ie_-,\au_-).$$

Any function $\psi \colon (\Ge,\au) \rightarrow \C$ can be written as 
\begin{eqnarray*}
\psi(x_j)= \psi_j(x), & \mbox{where} & \psi_j \colon I_j \rightarrow \C
\end{eqnarray*}
with
\begin{equation*}
I_j= \begin{cases} [0,a_j], & \mbox{if} \ j\in \Ie, \\ [0,\infty), &\mbox{if} \ j\in \Ee.  \end{cases}
\end{equation*}
One defines 
\begin{equation*}
\int_{\Ge} \psi := \sum_{i\in \Ie}  \int_{0}^{a_i}  \psi(x_i) \, dx_i +  \sum_{e\in \Ee}  \int_{0}^{\infty}  \psi(x_e)\, dx_e, 
\end{equation*}
where $dx_{i}$ and $dx_e$ refers to integration with respect to the Lebesgue measure on the intervals $[0,a_i]$ and $[0,\infty)$, respectively. 

By $\We_j$ with $j\in \Ee \cup \Ie$ one denotes the set of all $\psi_j\in \He_j$ such that its derivative $\psi_j$ is absolutely continuous and $\psi_j^{\prime}$ is square integrable. On the whole metric graph one considers the direct sum 
\begin{eqnarray*}
\We\equiv\We(\Ee,\Ie)= \bigoplus_{j\in \Ee \cup \Ie} \We_j.
\end{eqnarray*}

By $\De_j$ with $j\in \Ee \cup \Ie$ denote the set of all $\psi_j\in \He_j$ such that $\psi_j$ and its derivative $\psi_j^{\prime}$ are absolutely continuous and $\psi_j^{\prime\prime}$ is square integrable. Let $\De_j^0$ denote the set of all elements $\psi_j\in \De_j$ with
\begin{eqnarray*}
\psi_j(0)=0, &  \psi^{\prime}(0)=0, & \mbox{for}  \ j\in \Ee,    \\
\psi_j(0)=0, &  \psi^{\prime}(0)=0, & \psi_j(a_j)=0, \  \psi^{\prime}(a_j)=0, \ \mbox{for} \ j\in \Ie. 
\end{eqnarray*}
In the space $\He$ one defines the subspaces 
\begin{eqnarray*}
\De= \bigoplus_{j\in \Ee \cup \Ie} \De_j, & \mbox{and} &  \De^0= \bigoplus_{j\in \Ee \cup \Ie} \De_j^0.
\end{eqnarray*}

Note that these spaces clearly decouple the edges of the graph. The coupling is going to be implemented in terms of boundary conditions. With the scalar products $\langle \cdot,\cdot \rangle_{\We}$ defined by 
\begin{eqnarray*}
\langle \varphi,\psi \rangle_{\We}= \langle \varphi,\psi \rangle_{\He} + \langle \varphi^{\prime},\psi^{\prime} \rangle_{\He} 
\end{eqnarray*}
and $\langle \cdot,\cdot \rangle_{\De}$ given by 
\begin{eqnarray*}
\langle \varphi,\psi \rangle_{\De}= \langle \varphi,\psi \rangle_{\He} + \langle \varphi^{\prime},\psi^{\prime} \rangle_{\He} + \langle \varphi^{\prime\prime},\psi^{\prime\prime} \rangle_{\He}
\end{eqnarray*}
the spaces $\We$ and $\De$, respectively become themselves Hilbert spaces.

\subsection{Differential operator} 
in this article the formal differential operator  
\begin{align}\label{tau}
\left( \tau \psi\right)_j (x) = \begin{cases} -\frac{d^2}{dx}\psi_j(x), & j\in \Ee_+\cup \Ie_+, \  x\in I_j, \\ 
                                              +\frac{d^2}{dx}\psi_j(x), & j\in \Ee_-\cup \Ie_-, \  x\in I_j
                                \end{cases}
\end{align}
is considered. That is $\tau$ acts as $-\Delta$ on $(\Ge_+,\au_+)$ and as $+\Delta$ on $(\Ge_-,\au_-)$, where  
\begin{align}\label{Delta}
\left( \Delta \psi\right)_j (x) = \frac{d^2}{dx}\psi_j(x), && j\in \Ee\cup \Ie, \  x\in I_j. 
\end{align}
To emphasise the importance of the domains and to make the choice of the domains transparent, the notation is used to write a differential operator as 2-tuple, consisting of a differential expression and a domain. The operator 
$$T^{\min}=(\tau, \De_0)$$ 
is closed, densely defined and symmetric. Its adjoint in the Hilbert space $\He$ is the operator $T^{\max}=\left(T^{\min}\right)^*$,
$$T^{\max}=(\tau, \De).$$ 
Since the operator $T^{\min}$ has equal deficiency indices 
\begin{eqnarray*}
(d_+,d_-)=(d,d), & \mbox{where }  d=n+m,
\end{eqnarray*}\index{$T^{\min}$, $T^{\max}$}
there are self-adjoint realizations $T=T^*$ of $\tau$ with 
$$T^{\min} \subset T \subset T^{\max},$$ 
and all these realizations can be parametrized using boundary conditions.

\subsection{Space of boundary values}
On the lines of \cite{VKRS1999} one introduces the auxiliary Hilbert spaces 
\begin{equation*}\index{$\Ke_{\pm}$}
\Ke_{\pm} \equiv \Ke(\Ee_{\pm}, \Ie_{\pm}) = \Ke_{\Ee_{\pm}}  \oplus \Ke_{\Ie_{\pm}}^- \oplus \Ke_{\Ie_{\pm}}^+
\end{equation*}
with $\Ke_{\Ee_{\pm}} \cong \C^{\abs{\Ee_{\pm}}}$ and $\Ke_{\Ie_{\pm}}^{(\pm)} \cong \C^{\abs{\Ie_{\pm}}}$. Eventually one defines  
\begin{eqnarray*}
\Ke= \Ke_+ \oplus \Ke_-&  \mbox{and} & \Ke^2=\Ke \oplus \Ke.
\end{eqnarray*}
For $\psi\in \De$ one defines the vectors of boundary values
\begin{eqnarray*}\index{$\underline{\psi}_{\pm}$, $\underline{\psi}_{\pm}^{\prime}$}
\underline{\psi}_+=\begin{bmatrix} \{\psi_j(0)\}_{j\in \Ee_+} \\ \{\psi_j(0)\}_{j\in \Ie_+}\\ \{\psi_j(a_j)\}_{j\in \Ie_+}\end{bmatrix}, &
\underline{\psi}_-=\begin{bmatrix} \{\psi_j(0)\}_{j\in \Ee_-} \\ \{\psi_j(0)\}_{j\in \Ie_-}\\ \{\psi_j(a_j)\}_{j\in \Ie_-}\end{bmatrix}, &
\underline{\psi}=\begin{bmatrix} \underline{\psi}_+ \\ \underline{\psi}_-\end{bmatrix}, \\
\underline{\psi^{\prime}}_+= \begin{bmatrix} \{\psi_j^{\prime}(0)\}_{j\in \Ee_+} \\ \{\psi_j^{\prime}(0)\}_{j\in \Ie_+}\\ \{-\psi_j^{\prime}(a_j)\}_{j\in \Ie_+} \end{bmatrix}, &
\underline{\psi^{\prime}}_-= \begin{bmatrix} \{\psi_j^{\prime}(0)\}_{j\in \Ee_-} \\ \{\psi_j^{\prime}(0)\}_{j\in \Ie_-}\\ \{-\psi_j^{\prime}(a_j)\}_{j\in \Ie_-} \end{bmatrix}, &
\underline{\psi^{\prime}}= \begin{bmatrix} \underline{\psi^{\prime}}_+ \\ \underline{\psi^{\prime}}_- \end{bmatrix}.
\end{eqnarray*}
The vector $[\psi]= \underline{\psi} \oplus \underline{\psi^{\prime}}$ is an element of $\Ke^2 = \Ke\oplus \Ke$.

\section{Self--adjoint realizations}\label{secsa}

To measure how far the maximal operator $T^{\max}$ is away from being self-adjoint one evaluates the form $\jm_{(\Ge_+,\Ge_-)}$ defined by 
\begin{eqnarray*}
 \jm_{(\Ge_+,\Ge_-)}(\psi,\varphi):= \langle T^{\max} \psi ,\varphi \rangle - \langle  \psi , T^{\max} \varphi\rangle, & \mbox{for } \psi,\varphi\in \De,
\end{eqnarray*}
compare for example \cite[Section 3]{Harmer}. Note that this form can be represented by a matrix in the finite dimensional Hilbert space $\Ke^2$, that is 
$$\jm_{(\Ge_+,\Ge_-)}(\psi,\varphi)= \langle [\psi],\Jm_{(\Ge_+,\Ge_-)}[\varphi]\rangle_{\Ke^2},$$ where the matrix
\begin{align*}
\Jm_{(\Ge_+,\Ge_-)}= \begin{bmatrix} 0 &  0 &  -\mathds{1}_n & 0\\ 0 & 0 &  0 & \mathds{1}_{m}  \\ \mathds{1}_n & 0 & 0 & 0  \\ 0 & -\mathds{1}_{m} & 0 & 0 \end{bmatrix}, 
\end{align*}
is written with respect to the decomposition of $\Ke^2=\Ke_+\oplus \Ke_- \oplus \Ke_+ \oplus \Ke_-$; $\mathds{1}_n$ denotes the identity operator in the $n$-dimensional space $\Ke_+$ and $\mathds{1}_m$ the identity in the $m$-dimensional space $\Ke_-$, respectively. By an abuse of notation one denotes by $\jm_{(\Ge_+,\Ge_-)}$ also the sesquilinear form defined by $\Jm_{(\Ge_+,\Ge_-)}$ on $\Ke^2$. Notice that 
\begin{eqnarray*}
\Jm_{(\Ge_+,\Ge_-)}^*=-\Jm_{(\Ge_+,\Ge_-)} & \mbox{and} & \Jm_{(\Ge_+,\Ge_-)}^2= -\mathds{1}_{\Ke^2}
\end{eqnarray*}
holds and therefore $\Jm_{(\Ge_+,\Ge_-)}$ is a Hermitian symplectic matrix in $\Ke^2$. The relation of $\Jm_{(\Ge_+,\Ge_-)}$ to the standard Hermitian symplectic matrix 
\begin{align*}
\Jm_{\Ge}= \begin{bmatrix} 0 &  0 &  -\mathds{1}_n & 0\\ 0 & 0 &  0 & -\mathds{1}_{m}  \\ \mathds{1}_n & 0 & 0 & 0  \\ 0 & \mathds{1}_{m} & 0 & 0 \end{bmatrix}, 
\end{align*}
is provided by the similarity relation
\begin{eqnarray}\label{lemmasabc}
\Jm_{\Ge}=H_{n,m}\Jm_{(\Ge_+,\Ge_-)}H_{n,m}, & \mbox{where } H_{n,m}=\begin{bmatrix} \mathds{1}_n &  0 &  0 & 0\\ 0 & \mathds{1}_{m} & 0 & 0  \\ 0 & 0 & \mathds{1}_n & 0  \\ 0 & 0 & 0 & -\mathds{1}_{m}   \end{bmatrix}.
\end{eqnarray}
A subspace $\Me\subset\Ke^2$ is called \textit{maximal isotropic with respect to $\jm_{(\Ge_+,\Ge_-)}$} if 
\begin{eqnarray*}
\jm_{(\Ge_+,\Ge_-)}(\xi_1,\xi_2) = \left\langle \xi_1 , \Jm_{\Ge_+,\Ge_-} \xi_2 \right\rangle=0 & \mbox{for all} & \xi_1,\xi_2\in \Me
\end{eqnarray*}
and if there are no proper subspaces $\Me\subsetneq\Me^{\prime}$ having this property. Any subspace $\Me\subset \Ke^2$ with $\dim \Me\geq d$, can be parametrized as follows. Let $A$ and $B$ be linear maps in $\Ke$. By $(A, \, B)$ one denotes the linear map from $\Ke^2=\Ke \oplus \Ke$ to $\Ke$ defined by $(A, \, B) (\eta_1 \oplus \eta_2) = A\eta_1 + B \eta_2,$ where $\eta_1,\eta_2\in\Ke$. One sets 
\begin{eqnarray*}
\Me=\Me(A,B), & \mbox{where} & \Me(A,B):=\Ker (A, \, B).
\end{eqnarray*}
From \eqref{lemmasabc} one deduces a one--to--one correspondence between subspaces in $\Ke^2$ which are maximal isotropic with respect to the standard Hermitian symplectic form, which is well studied, and subspaces in $\Ke^2$ that are maximal isotropic with respect to the Hermitian symplectic form $\jm_{(\Ge_+,\Ge_-)}$ occurring here.

\begin{theorem}\label{sabc}
All self-adjoint extensions of $T^{\min}$ are uniquely determined by subspaces $\Me\subset \Ke^2$, which are maximal isotropic with respect to the Hermitian symplectic form $\jm_{(\Ge_+,\Ge_-)}$. All such maximal isotropic subspaces are given by $\Me=\Me(A,B)$, where $A$ and $B$ are linear maps in $\Ke$, which satisfy the two conditions
\begin{enumerate}\index{$J_{n,m}$}
\item $\mbox{Rank} \, (A, \, B)= n+m$,  that is the rank is maximal and 
\item $B J_{n,m} A^{\ast}=A J_{n,m} B^{\ast}$, where $J_{n,m}:=\begin{bmatrix}  \mathds{1}_n &  0 \\ 0 & -\mathds{1}_{m} \end{bmatrix} $ defines a map in \\ $\Ke=\Ke_+\oplus \Ke_-$.
\end{enumerate}          
Any self--adjoint extension of $T^{\min}$ is given by $T(A,B)=(\tau,\Dom (T(A,B)))$ with
\index{$T(A,B)$}
$$\Dom (T(A,B))=\left\{ \psi\in \De \mid A\underline{\psi}+ B\underline{\psi^{\prime}} =0 \right\},$$ 
where $A,B$ satisfy both conditions (1) and (2). 
\end{theorem}

\begin{remark}\ \ \ \ \ \ \ \  \
\begin{enumerate}
\item The condition $A\underline{\psi}+ B\underline{\psi^{\prime}} =0$ is equivalent to $[\psi]\in \Me(A,B)$.
\item The operator $-T(A,B)$ is the operator with positive and negative edges interchanged, but the coupling is implemented by the same boundary conditions at the vertices. 
\item For $\Ge_{\mp}=\emptyset$ the operator $T(A,B)$ is the self-adjoint operator plus or minus Laplace $\mp\Delta(A,B)$ on $\Ge_{\pm}$, where $\Delta(A,B)$ denotes the realization of $\Delta$ with domain $\Dom(\Delta(A,B))=\{\psi\in \De \mid A\underline{\psi} + B \underline{\psi}^{\prime}=0 \}$, compare for example \cite{VKRS1999}.
\item Observe that the second condition in Theorem~\ref{sabc} can be re--formulated in terms of relations in finite dimensional Krein spaces. Consider the space $\Ke$ equipped with the indefinite inner product $[\cdot,\cdot]=\langle \cdot ,J_{n,m}\cdot \rangle_{\Ke}$. Then $B J_{n,m} A^{\ast}=B A^{[*]}$, where $\cdot^{[*]}$ denotes the Krein space adjoint.
\end{enumerate}
\end{remark}

\begin{proof}[Proof of Theorem~\ref{sabc}]
It is well known that all self-adjoint extensions of $T^{\min}$, which is a closed symmetric operator with finite and equal deficency indices, is determined by subspaces $\Me$ of $\Ke^2$, which are maximal isotropic with respect to the Hermitian symplectic form $\jm_{(\Ge_+,\Ge_-)}$, see for example \cite{Harmer} and the references therein.  

To prove the parametrization let now $\Me$ be a maximal isotropic subspace of $\Ke^2$ with respect to the  Hermitian symplectic form defined by $\Jm_{(\Ge_+,\Ge_-)}$. Then by \eqref{lemmasabc} the space $H_{n,m}\mathcal{M}$ is a maximal isotropic subspace of $\Ke^2$ with respect to the standard Hermitian symplectic form defined by $\Jm_{\Ge}$. The space $H_{n,m}\mathcal{M}$ can be parametrized by matrices $A_0,B_0$, see for example \cite{VKRS1999}, which satisfy the two conditions
\begin{enumerate}
\item $\mbox{Rank} (A_{0}, \, B_{0})=n+m$ and 
\item $A_{0}B_{0}^{*}=B_{0}A_{0}^{*}$.
\end{enumerate}  
Hence the original space $\Me$ can be represented as  
$$\mathcal{M}=H_{n,m} \Ker \left(A_{0}, \, B_{0} \right) = \Ker \left(A_0, \, B_0 \begin{bmatrix} \mathds{1}_n &  0 \\ 0 & -\mathds{1}_{m} \end{bmatrix} \right).$$ 
The matrices 
\begin{eqnarray*}
A= A_0 & \mbox{and} & B= B_0 \begin{bmatrix} \mathds{1}_n &  0 \\ 0 & -\mathds{1}_{m} \end{bmatrix}
\end{eqnarray*}
satisfy the two conditions formulated in Theorem~\ref{sabc}. 
\end{proof}

\begin{remark}\label{UVK}
The proof of Theorem~\ref{sabc} shows that if one has a self--adjoint Laplace operator $-\Delta(A_0,B_0)$ parametrized by matrices $A_0$ and $B_0$, then $T(A,B)$ is self--adjoint with $A=A_0$ and $B=B_0 J_{n,m}$. In particular there is a unique parametrization in terms of unitary operators in $\Ke$. For Laplace operators it is known that for self--adjoint boundary conditions defined by operators $A_0$ and $B_0$ in $\Ke$ there exists a unique unitary map $U$ in $\Ke$ such that equivalent boundary conditions  are defined by 
\begin{eqnarray*}
A^{\prime}_0 = -\frac{1}{2}\left(U- \mathds{1} \right) &\mbox{and} & B^{\prime}_0 = \frac{1}{2i}\left(U+ \mathds{1} \right), 
\end{eqnarray*} 
see for example \cite[Section 3]{Harmer}. Hence for $T(A,B)$ self--adjoint one can give an equivalent parametrization in terms of the same unitary matrix $U$ with
\begin{eqnarray*}
A^{\prime} = -\frac{1}{2}\left(U- \mathds{1} \right) &\mbox{and} & B^{\prime} = \frac{1}{2i}\left(U+ \mathds{1} \right)J_{n,m}. 
\end{eqnarray*} 
\end{remark}

The parametrization by matrices $A$ and $B$ in Theorem~\ref{sabc} is not unique, since two operators $T(A^{\prime},B^{\prime})$ and $T(A,B)$ are equal if and only if the corresponding maximal isotropic subspaces $\Me(A^{\prime},B^{\prime})$ and $\Me(A,B)$ agree. Using the result from \cite[Theorem 6]{PKQG1}, for self-adjoint Laplacians on finite metric graphs one obtains a unique parametrization.

\begin{corollary}\label{PL}
Let $\Me(A,B)\subset \Ke^2$ be maximal isotropic with respect to the Hermitian symplectic form $\jm_{(\Ge_+,\Ge_-)}$. Then there exists a unique orthogonal projection $P$ and a unique Hermitian operator $L$ with $P^{\perp}L P^{\perp}=L$, where $P^{\perp}=\mathds{1}-P$,  such that with
\begin{align*}
A^{\prime}  = L + P, && B^{\prime} = P^{\perp} \begin{bmatrix}  \mathds{1}_n &  0 \\ 0 & -\mathds{1}_{m} \end{bmatrix}
\end{align*}
one has $\Me(A,B)= \Me(A^{\prime},B^{\prime})$, where $P$ is the orthogonal projector onto \\ $\Ker B J_{n,m}\subset \Ke$. 
\end{corollary}

\begin{example}\label{ex1IndQG}
Let $\Ge$ be the star graph consisting of two external edges, $\Ee_+=\{1\}$ and $\Ee_-=\{2\}$, glued together at one single vertex $\partial(1)=\partial(2)$, and consider 
\begin{eqnarray*}
A=\begin{bmatrix} -1 & 1 \\ 0 & 0\end{bmatrix}&\mbox{and} & B=\begin{bmatrix} 0 & 0 \\ -1 & 1\end{bmatrix}.
\end{eqnarray*}
Since these matrices satisfy the two conditions formulated in Theorem~\ref{sabc} the operator $T(A,B)$ defined on $\Ge$ is self-adjoint. One can identify the metric graph $\Ge$ with the real line, and under this identification the operator $T(A,B)$ corresponds to the operator $-\tfrac{d}{dx} \sgn  (x) \tfrac{d}{dx}$ defined in $L^2(\R)$ on its natural domain. 
\end{example}

From the self--adjointness of $T(A,B)$ one deduces that the eigenvalue equation  can have square integrable solutions only for real spectral parameter. The self-adjoint operator $T(A,B)$ can be interpreted as a realization of a system of the type 
\begin{eqnarray*}
\tau=\begin{bmatrix}
-\Delta & 0 \\ 0 & +\Delta
\end{bmatrix}
\end{eqnarray*}
which is given with respect to the decomposition $\Ee \cup \Ie=\left(\Ee_+ \cup \Ie_+ \right)\cup \left(\Ee_- \cup \Ie_- \right)$. The operator $T(A,B)$ generates a unitary group $U(t)=e^{-it T(A,B)}$ which provides solutions for the initial value problem 
\begin{equation*}
\left\{ \begin{array}{ll}
         \left( i \frac{\partial}{\partial t } - T(A,B) \right) u(x,t)=0, \qquad t\in \mathbb R,\\ 
        u(\cdot,0)=u_0. \end{array} \right. 
\end{equation*}
Compared to the time--dependent behaviour of groups generated by (positive) self--adjoint Laplacians $-\Delta$, the group generated by $T(A,B)$ describes a time--reversed behaviour on the negative edges and the forward--directed behaviour on the positive edges. 

\section{Self-adjointness in Krein spaces}
One can consider the same graph described in Example~\ref{ex1IndQG}, but one imposes boundary conditions defined by  
\begin{eqnarray*}
A=\begin{bmatrix} -1 & 1 \\ 0 & 0\end{bmatrix} &\mbox{and} & B=\begin{bmatrix} 0 & 0 \\ 1 & 1\end{bmatrix}.
\end{eqnarray*}
These define the non--self--adjoint operator $T(A,B)$ and the self--adjoint operator $-\Delta(A,B)$, compare for example \cite{VKRS1999}. Here $-\Delta(A,B)$ denotes the Laplace operator defined in \eqref{Delta} with domain $\Dom(\Delta(A,B))=\{\psi \in \De \mid A \underline{\psi} + B\underline{\psi}^{\prime}=0\}$. The operator $T(A,B)$ corresponds, again by identifying this graph with the real line, to the non--self--adjoint operator $-\sgn  (x)\tfrac{d^2}{dx^2}$ with domain $H^2(\R)\subset L^2(\R)$, where $H^2(\R)$ denotes the Sobolev space of order two. A generalization of $-\sgn  (x)\tfrac{d^2}{dx^2}$ from intervals to finite metric graphs can be done on the same lines as for $-\tfrac{d}{dx}\sgn(x)\tfrac{d}{dx}$. Define the Krein space
\begin{eqnarray*}
\mathfrak{K}= \bigoplus_{j\in \Ee \cup \Ie} L^2(I_j), & \mbox{with indefinite inner product}  & [\cdot,\cdot]_{\mathfrak{K}}= \langle \cdot, J_{n,m}\cdot \rangle_{\He}, 
\end{eqnarray*}
where with a slight abuse of notation one sets
\begin{equation*}
(J_{n,m} \psi)_j = \begin{cases} +\psi_j, & j\in \Ee_+ \cup \Ie_+, \\ -\psi_j, & j\in \Ee_- \cup \Ie_-. \end{cases}
\end{equation*}
The multiplication with $J_{n,m}$ is the fundamental symmetry of the Krein space $\mathfrak{K}$. Since 
\begin{eqnarray*}
J_{n,m} T^{\max}= -\Delta & \mbox{and} & J_{n,m} T^{\min}= -\Delta^0
\end{eqnarray*}
holds, where $\Delta^0$ is the restriction of $\Delta$ to $\De^0$, one obtains
\begin{proposition}\label{krein}
The extensions of $T^{\min}$ in $\He$ with $J_{n,m}T(A,B)$ self-adjoint are defined exactly by those $A,B$ for which $-\Delta(A,B)$ is self-adjoint in $\He$. 
\end{proposition}
As already mentioned the self-adjoint realizations $-\Delta(A,B)$ of the maximal Laplacian $-\Delta$ are characterized completely, see for example \cite{VKRS1999}. The eigenvalue problem for operators of the type $-\sgn  (x)\tfrac{d^2}{dx^2}$ on compact finite metric graphs, which have been described here in Proposition~\ref{krein}, is considered in \cite{Currie2011}.

\begin{remark}\label{sgndxsgndx}
Analogue to the case treated in Proposition~\ref{krein} generalizations of the operator $$-\sgn  (x)\frac{d}{dx}\sgn(x)\frac{d}{dx}$$ from intervals to metric graphs can be given. One can search for those extensions $-\Delta(A,B)$ of the minimal Laplacian $-\Delta^0$ with $-J_{n,m}\Delta(A,B)$ self-adjoint in $\He$ (or equivalently $-\Delta(A,B)$ Krein-space-self-adjoint in $\mathfrak{K}$). Analogue to Proposition~\ref{krein} one obtains that the operator $-J_{n,m}\Delta(A,B)$ is self-adjoint in $\He$ if and only if $T(A,B)$ is self-adjoint in $\He$. Consequently these extensions can be parametrized using Theorem~\ref{sabc}. 
\end{remark}

\section{Indefinite form methods}\label{secform}
An approach that has turned out to be very fruitful for sign--definite operators in Hilbert spaces is the one using quadratic forms. Representation theorems give a one to one correspondence between closed semi--bounded forms and semi--bounded self-adjoint operators, see for example \cite[Chapter VI]{Kato}. The quadratic form defined by $-\frac{d}{dx}\sgn(x)\frac{d}{dx}$ is symmetric, but not semi--bounded. However there are generalizations of the first representation theorem to indefinite quadratic forms. One formulation is given in \cite{GKMV}, for further information on indefinite quadratic forms see also the references given therein, in particular the works \cite{deSnoo} and \cite{McIntosh}. The application of results about indefinite quadratic forms to certain indefinite second order differential operators on bounded domains is going to be discussed in \cite{VK12}. An analogue construction can be given for certain compact metric graphs. Here the result is stated and compared to the previously obtained results, where methods from extension theory have been used. 

Assume that $(\Ge,\au)$ is a \textit{compact star graph} that is one has finitely many internal edges $\Ie$ and the initial vertices of all edges are unified in one vertex and all terminal vertices are vertices of degree one. This means $\partial_-(i_p)= \partial_-(i_q)$ for $i_p,i_q\in \Ie$ and $\partial_+(i_p)\neq\partial_+(i_q)$, for $i_p\neq i_q$. On the compact star graph $(\Ge,\au)$ one considers the spaces $H^1(\Ge)\subset \We$ and $H_0^1(\Ge)\subset \We$, where
$$H^1(\Ge):=\left\{ \psi \in \We \mid \psi_j(0)=\psi_i(0), \ i,j \in \Ie  \right\} $$ and $$ H_0^1(\Ge):=\left\{ \psi \in H^1(\Ge) \mid \psi_j(a_j)=0, \ j \in \Ie  \right\}.$$
Define now the gradient operator in $\He$ by
\begin{eqnarray*}
D\colon H_0^1(\Ge) \rightarrow \He, &\psi \mapsto \psi^{\prime}, 
\end{eqnarray*}
whose adjoint operator is 
\begin{eqnarray*}
D^{\ast}\colon H^1(\Ge) \rightarrow  \He, & \psi \mapsto -\psi^{\prime}. 
\end{eqnarray*}
The operator $D$ and the space $\Ran D\subset \He$ are closed. Since $\Ran D$ is closed, it is itself a Hilbert space with the Hilbert space structure inherited from $\He$. Hence   
$$Q\colon \He \rightarrow \Ran D,  \ \ Qu=\begin{cases}  u, & u\in \Ran D, \\ 0, & u\perp \Ran D \end{cases}$$
is a partial isometry. The adjoint $Q^*$ is the embedding from $\Ran D$ to $\He$. Note that $(\Ran D)^{\perp}$ is the space of constant functions and therefore the operator $Q$ becomes more precisely
\begin{eqnarray*}
Qu =  u - \frac{1}{\sum_{i\in\Ie}a_i} \int_{\Ge} u,
\end{eqnarray*}
since $Q$ maps onto the orthogonal complement of the constant functions.

The operators $D$ and $D^{\ast}$ were used implicitly in \cite{PKQG1} to apply form methods to Laplace operators on finite metric graphs. The following theorem is a graph version of the result that is going to be presented in \cite{VK12}.

\begin{theorem}\label{introthm:1IndQG}
Let $(\Ge,\au)$ be a compact star graph and 
\begin{itemize}
\item[(a)] $A\in L^\infty(\Ge;\R)$, that is $A_j\in L^\infty(I_j;\R)$, for each $j\in \Ie$, be such that
\item[(b)] the operator $Q M_A Q^\ast:\Ran D\to\Ran D$ is boundedly invertible, where $M_A$ is the multiplication operator 
$$M_A\colon \He \rightarrow \He, \quad \phi \mapsto A(\cdot)\phi.$$
\end{itemize}
Then
\begin{itemize}
\item[(i)] there exists a unique self-adjoint operator $\cL$ with $\Dom(\cL)\subset H_0^1(\Ge)$ such that
\begin{equation*}
 \langle \varphi, \cL \psi\rangle_{\He} = \langle \varphi^{\prime}, A(\cdot) \psi^{\prime}\rangle_{\He}
\end{equation*}
holds for all $\varphi\in H_0^1(\Ge)$ and all $\psi\in\Dom(\cL)$, where
\begin{equation*}
\Dom(\cL)=\{\psi\in H_0^1(\Ge)\, |\, M_A D \psi\in H^1(\Ge)\}.
\end{equation*}
For any $\psi\in\Dom(\cL)$ one has $\cL \psi = D^\ast M_A D\psi$, the domain $\Dom(\cL)$ is a core for the gradient operator $D$;
\item[(ii)] the operator $\cL$ is boundedly invertible and its inverse $\cL^{-1}$ is compact. In particular, the spectrum of $\cL$ is purely discrete.
\end{itemize}
\end{theorem}

Roughly speaking, Theorem~\ref{introthm:1IndQG} states that the operator $\cL$ is a self--adjoint realization of the formal differential expression $-\tfrac{d}{dx} A(\cdot)\tfrac{d}{dx}$ with Dirichlet boundary conditions. The proof of Theorem~\ref{introthm:1IndQG} relies on the representation theorem for indefinite quadratic forms \cite[Theorem 2.3 and Lemma 2.2]{GKMV}. It is completely analogue to the one that is going to be given in \cite{VK12} and omitted here. The condition (b) of Theorem~\ref{introthm:1IndQG} can be verified by
\begin{proposition}\label{QAQIndefQG}
Let $(\Ge,\au)$ be a compact star graph and $A(\cdot), A(\cdot)^{-1}\in L^\infty(\Ge;\R)$. Then $QAQ^*$ is boundedly invertible if and only if
$$\displaystyle{\int_{\Ge} \frac{1}{A} \neq 0.}$$
\end{proposition}
The proof is based on the fact that the space $(\Ran D)^{\perp}= \Ker D^{\ast}$ is the one dimensional space of constant functions in $H^1(\Ge)$. It is analogue to the one that is going to be presented in \cite{VK12} and omitted here. Applying Proposition~\ref{QAQIndefQG} to the function 
$$J_{n,m}\colon \Ge \rightarrow \C, \ x_j \mapsto \begin{cases} +1, & j\in \Ie_+, \\ -1, & j\in \Ie_-,  \end{cases}$$ 
where $\Ie= \Ie_+ \dot\cup \Ie_-$ one obtains that $QM_{J_{n,m}}Q^*$ is boundedly invertible if and only if
$$\displaystyle{\sum_{i\in \Ie_+}a_i - \sum_{i\in \Ie_-}a_i \neq 0}.$$ 
Then by applying Theorem~\ref{introthm:1IndQG} one obtains that the form $\mathfrak{l}_{n,m}$ given by 
\begin{eqnarray*}
\mathfrak{l}_{n,m}[\varphi,\psi]=  \langle \varphi^{\prime}, J_{n,m}\psi^{\prime} \rangle_{\He}, & & \varphi,\psi\in H^1_0(\Ge)\subset \He
\end{eqnarray*}
defines uniquely the operator $\cL_{n,m}=D^* J_{n,m} D$ in $\He$ with natural domain \\ $\Dom(\cL_{n,m})\subset H_0^1(\Ge)$. Note that the operator $\cL_{n,m}$ is a self--adjoint extension of $T^{\min}$ and therefore there are operators $A$ and $B$ in $\Ke$ such that $\cL_{n,m}=T(A,B)$. More precisely, the operator $T(A,B)$ is defined by Dirichlet boundary conditions on the vertices of degree one and at the central vertex $\nu=\partial_-(i_p)$, $i_p\in \Ie$ by the local boundary conditions that are given by the matrices
\begin{align*}
A_{\nu}= \left[
   \begin{array}{cccccc}
     1 & -1 & 0 &\cdots & 0 & 0 \\
     0 & 1 & -1 &\cdots & 0 & 0  \\
     0 & 0 & 1 &\cdots & 0 & 0  \\
     \vdots &\vdots  & \vdots & & \vdots  & \vdots \\
        0 & 0 & 0 &\cdots & 1 & -1  \\
   0 & 0 & 0 &\cdots & 0 & 0 
   \end{array}
\right], && B_{\nu}= \left[
   \begin{array}{cccccc}
     0 & 0 & 0 &\cdots & 0 & 0 \\
     0 & 0 & 0 &\cdots & 0 & 0  \\
     0 & 0 & 0 &\cdots & 0 & 0  \\
     \vdots &\vdots  & \vdots & & \vdots  & \vdots \\
        0 & 0 & 0 &\cdots & 0 & 0  \\
   1 & \cdots & 1 & \cdots & -1 & -1 
   \end{array}
\right],
\end{align*}  
where in the last row of $B_{\nu}$ for each edge in $\Ie_+$ stands a $+1$ and for each edge in $\Ie_-$ a $-1$. To paraphrase, these local boundary conditions guarantee that functions are continuous at the central vertex and that the sum of the outward directed derivatives evaluated at the positive incident edges equals the sum of the outward directed derivatives evaluated at the negative incident edges. These local boundary conditions arise naturally from the form approach and therefore they are the main example for self-adjoint boundary conditions discussed here. Note that these boundary conditions are related to the so--called standard or Kirchhoff boundary conditions, compare for example \cite{PKQG1}. If $A_{\nu}^{st}$ and $B_{\nu}^{st}$ define the standard boundary conditions at the central vertex then one has $A_{\nu}=A_{\nu}^{st}$ and $B_{\nu}= B_{\nu}^{st} J_{n,m}$.  

\begin{remark}
The form approach admits straight forward generalizations to compact graphs, as long as the operator $D\colon H_0^1(\Ge) \rightarrow \He$ is boundedly invertible. This is equivalent to the invertibility of the Dirichlet Laplacian $-\Delta_D=D^{\ast}D$. Hence the form approach applies whenever there is at least one vertex of degree one. The obstacle that the form approach cannot be used for arbitrary not necessarily compact finite metric graphs seems to be  due to technical difficulties.
\end{remark}

\section{Extension theory background}\label{secext}
One can make one step back and have a look at the problem from the more general viewpoint of extension theory. Extension theory deals with self-adjoint extensions of closed symmetric operators in Hilbert spaces. Actually it has been the starting point of this work to observe that the operator $T^{\min}$ is a closed symmetric operator with equal deficiency indices. 

\subsection{Classical extension theory}
The classical theory developed by von Neumann describes  self--adjoint extensions of closed symmetric operators in terms of unitary mappings between the deficiency spaces, see for example \cite{Faris}. Instead of the concrete operator $T^{\min}$ consider now two arbitrary closed symmetric operators $A_+$ and $A_-$ in Hilbert spaces $\He_+$ and $\He_-$, respectively, each with equal deficiency indices 
\begin{eqnarray*}
d_{+}(A_+)=d_{-}(A_+)\,(\leq \infty), & d_{+}(A_-)=d_{-}(A_-)\,(\leq \infty),
\end{eqnarray*}
where $d_{\pm}(A_{\pm})= \dim \Ker \left( A^{\ast}\mp i \right)$. Then one can study two operators
\begin{eqnarray*}
T:= A_+ \oplus -A_- & \mbox{and} & \Delta:= A_+ \oplus A_-.
\end{eqnarray*}
Both are  closed operators in $\He=\He_+ \oplus \He_-$ with equal deficiency indices
\begin{eqnarray*}
d_{\pm}(A)=d_{\pm}(A_+)+ d_{\mp}(A_-) & \mbox{and} & d_{\pm}(\Delta)=d_{\pm}(A_+)+ d_{\pm}(A_-).
\end{eqnarray*}
In von Neumann's theory the deficiency spaces 
\begin{eqnarray*}
N_{\pm}(z)=\Ker (A_{\pm}^*-z), & & z\in \C\setminus \R
\end{eqnarray*} 
are important objects. One considers $N_T(z):= \Ker( T^*-z)$. By construction 
\begin{align*}
N_T(z)=\Ker \begin{bmatrix} A_+^*-z & 0  \\ 0 & A_-^*+z \end{bmatrix}
\end{align*}
and therefore $N_T(z)= N_+(z) \oplus N_-(-z)$, whereas for $N_{\Delta}(z)=\Ker (\Delta^*-z)$ 
\begin{align*}
N_{\Delta}(z)=\Ker \begin{bmatrix} A_+^*-z & 0  \\ 0 & A_-^*-z \end{bmatrix}
\end{align*}
holds and therefore $N_{\Delta}(z)= N_+(z) \oplus N_-(z)$. The self-adjoint extensions of $T$ are in a one-to-one correspondence to the unitary mappings $$U^{\prime}\colon N_T(i)  \rightarrow N_T(-i).$$ Analogous to the above the self-adjoint extensions of $\Delta$ are in a one-to-one correspondence to the unitary mappings $$U \colon  N_{\Delta}(i) \rightarrow N_{\Delta}(-i).$$ One observes that $N_-(i)$ and $N_-(-i)$ are unitarily equivalent, because $A_-$ is a closed and symmetric operator with equal deficiency indices. Let $J\colon N_-(i) \rightarrow N_-(-i)$ be such a unitary equivalence. The relation between unitary mappings $U$ and $U^{\prime}$ is summarized in the following diagram
\begin{align*}
\begin{CD}
N_T(i)=N_{+}(i)\oplus N_{-}(i)   @>\mathds{1}\oplus J>> N_{+}(i)\oplus N_{-}(-i)= N_T(i)\\
@V U VV                  @V U^{\prime}VV \\
N_{\Delta}(i)=N_{+}(-i)\oplus N_{-}(-i)   @>\mathds{1}\oplus J>>  N_{+}(-i)\oplus N_{-}(i)= N_{\Delta}(i) 
\end{CD}
\end{align*}
and one reads from this that the relation is a one-to-one correspondence, this means each unitary map $U$ defines uniquely a unitary map $U^{\prime}$ and vice versa. This in turn gives a one--to--one correspondence between the self-adjoint extensions of $\Delta$ and the self-adjoint extensions of $T$. 
 
\subsection{Spaces of boundary values} In the context of differential operators the description of self-adjoint extensions in terms of boundary values can be more practical. Following A.~N.~Ko\v{c}ube\u\i, see \cite[Theorem 3]{Kochubei1974}, for each  symmetric operator $X$, defined in a Hilbertspace $\He$, with equal deficiency indices, $d_+(X)=d_-(X)\leq \infty$ there exists a Hilbert space $\Ke$ of dimension $d=d_+(X)$ and linear transformations $\Gamma^1,\Gamma^2\colon \Dom(X^*) \rightarrow \Ke $ with the properties,
\begin{enumerate}
\item for any $\varphi,\psi \in \Dom(X^*)$
$$\langle X^*\varphi,\psi \rangle_{\He}- \langle \varphi,X^*\psi \rangle_{\He}= \langle \Gamma^1\varphi,\Gamma^2\psi \rangle_{\Ke} - \langle \Gamma^2\varphi,\Gamma^1\psi \rangle_{\Ke},  $$
\item for any $\kappa_1,\kappa_2\in K$ there is a $\varphi\in \Dom(A^{\ast})$ such that $\Gamma^1\varphi=\kappa_1$ and $\Gamma^2\varphi=\kappa_2$, and  
\item if $\varphi\in \Dom(X)$, then $\Gamma^1\varphi=\Gamma^2\varphi=0$.
\end{enumerate}
This can be interpreted as a generalization of Green's formula or simply of integration by parts. The triple $(\Ke,\Gamma^1,\Gamma^2)$ is called \textit{space of boundary values} of $X$. The self-adjoint extensions of $X$ are given in terms of unitary mappings $U$ in $\Ke$, see \cite[Theorem 2 and 4]{Kochubei1974}. More precisely one has that all self--adjoint extensions $\widetilde{X}$ of $X$ are restrictions of $X^*$, which can be parametrized in terms of $U$, that is one has $\widetilde{X}=X_U$, where $X_U \psi = X^{\ast}\psi$ and
\begin{eqnarray*}
\Dom(X_U)= \{ \varphi \in \Dom(X^*) \mid (U-\mathds{1})\Gamma^1 \varphi +  i(U + \mathds{1})\Gamma^2 \varphi=0 \}.
\end{eqnarray*}

Consider the same situation as above with closed and symmetric operators $A_{\pm}$ each with equal deficiency indices. Let $(\Ke_{\pm},\Gamma^1_{\pm}, \Gamma^2_{\pm})$ be the spaces of boundary values of $A_{\pm}$. Then the space of boundary values for $\Delta$ is given by $(\Ke,\Gamma_{\Delta}^1, \Gamma_{\Delta}^2)$ with $\Ke= \Ke_+ \oplus \Ke_-$,
\begin{eqnarray*}
\Gamma_{\Delta}^1= \begin{bmatrix} \Gamma^1_+ & 0 \\ 0 & \Gamma^1_- \end{bmatrix} &\mbox{and} & \Gamma_{\Delta}^2= \begin{bmatrix} \Gamma^2_+ & 0 \\ 0 & \Gamma^2_- \end{bmatrix},
\end{eqnarray*}
which are given with respect to the decomposition $\Ke= \Ke_+ \oplus \Ke_-$. For $T$ the space of boundary values is given then by $(\Ke,\Gamma_{T}^1, \Gamma_{T}^2)$, again with $\Ke= \Ke_+ \oplus \Ke_-$, but now with 
\begin{eqnarray*}
\Gamma_{T}^1= \begin{bmatrix} \Gamma^1_+ & 0 \\ 0 & \Gamma^1_- \end{bmatrix} &\mbox{and} & \Gamma_{T}^2= \begin{bmatrix} \Gamma^2_+ & 0\\ 0 & -\Gamma^2_- \end{bmatrix}.
\end{eqnarray*}
As in the parametrization of von Neumann there is a one--to--one correspondence of the self-adjoint extensions of $\Delta$ and $T$. Given a unitary map $U$ in $\Ke$ then $\Delta_U$ is the restriction of $\Delta^*$ to 
\begin{eqnarray*}
\Dom(\Delta_U)= \{ \varphi\in \Dom(\Delta^*) \mid (U-\mathds{1})\Gamma_{\Delta}^1 \varphi +  i(U + \mathds{1})\Gamma_{\Delta}^2 \varphi=0 \}.
\end{eqnarray*}
With the same $U$ one obtains the extension of $T$ as a restriction of $T^*$ with domain 
\begin{eqnarray*}
\Dom(T_U)= \{ \psi\in \Dom(T^*) \mid (U-\mathds{1})\Gamma_{T}^1 \psi +  i(U + \mathds{1})\Gamma_{T}^2 \psi=0 \},
\end{eqnarray*}
where $\Dom(T^*)=\Dom(\Delta^*)$, compare Remark~\ref{UVK}.

\subsection{A radially symmetric example}
To come full circle one goes back to the starting point -- the differential operator given in \eqref{eq1IndefQG}. Consider the formal differential operator $\tau$ in $L^2(\R^2)$ defined by 
\begin{align*}
\tau u= -\div A(\cdot)\grad u, &&  A(x)= \begin{cases}+1, & x\in\Omega_+, \\  -1, & x\in\Omega_-,     \end{cases}
\end{align*}
where
\begin{align*}
&\Omega_+= \{x\in \R^2 \mid \norm{x}\leq R_1\} \cup \{x\in \R^2 \mid \norm{x}\geq R_2\} \  \mbox{and} \\
 &\Omega_-=\{x\in \R^2 \mid R_1 < \norm{x} <R_2   \}
\end{align*}
with $R_2>R_1>0$. This ring geometry has been studied in \cite{Bouchitte}. As $A(\cdot)$ is radially symmetric one can transform the operator $\tau$ to polar coordinates. Using the divergence and the gradient operator in polar coordinates one obtains
\begin{align*}
\tau u(r,\theta)= \frac{1}{r} \frac{d}{d r} r A(r) \frac{d}{d r} u(r,\theta) + \frac{1}{r^2} A(r) \frac{d^2}{d \theta^2}   u(r,\theta),
\end{align*}
where 
$$ A(r)= \begin{cases}+1, & r\in (0,R_1]\cup [R_3,\infty), \\  -1, & r\in (R_2,R_3).     \end{cases}$$
Using a separation of variables one obtains that $\tau$ is unitarily equivalent to 
\begin{align*}
\widetilde{\tau}=\bigoplus_{m\in \mathbb{Z}} \tau_m,
\end{align*}
where $\tau_m$ are operators in $L^2((0,\infty), dr)$ defined by 
\begin{align*}
 \tau_{m}u_m(r) = 
 -\frac{\partial r}{\partial r} A(r) \frac{\partial r}{\partial r} u_m(r) + A(r)\frac{m^2 -4^{-1}}{r}u_m(r).
\end{align*}
Denote by $AC_{loc}$ the set of locally absolutely continuous functions on the interval $(0,\infty)$. The natural (maximal) domain for $\tau_m$ in $L^2((0,\infty),dr)$ is 
\begin{align*}
\Dom(T_m)= \{ u_m\in L^2(0,\infty) \mid u_m, A(\cdot)u_m^{\prime} \in AC_{loc} \mbox{and} \ \tau_m u_m\in L^2((0,\infty),dr) \}
\end{align*}
for $m\neq 0$ and for $m=0$ one imposes additionally the assumption (or rather boundary condition)
$$\lim_{r\to 0}[\sqrt{r} \ln(r)]^{-1} u_0(r)= 0.$$ 
From the general considerations in this section or from the Sturm-Liouville theory, see for example \cite{WeidmannODE, Zettl}, one can deduce that the operators $$T_m=(\tau_m, \Dom(\tau_m))$$ are self--adjoint. Transforming back one obtains that $T=(\tau,\Dom(T))$ on the natural domain 
\begin{align*}
\Dom(T)= \left\{u \in L^2(\R^2) \mid  u \in H^1(\R^2), -\div A(\cdot)\grad u \in L^2(\Omega)    \right\}
\end{align*}
is self-adjoint. Observe that $\Dom(T)$ is a subset of $H^1(\R^2)$ and elements $u\in \Dom(T)$ satisfy the matching condition for the normal derivatives at the threshold of $\Omega_+$ and $\Omega_-$
$$ \lim_{\epsilon\to 0+}\frac{d}{d r} u(r+ \epsilon,\theta) = - \lim_{\epsilon\to 0-} \frac{d}{d r} u(r+\epsilon,\theta), $$
where $(r,\theta)\in \partial \Omega_+ =\partial \Omega_-$.

\section{Eigenvalues, resonances and resolvents}\label{secres}
The study of the spectral resolution of the self-adjoint operator $T(A,B)$ is based on finding solutions of $(\tau-k^2)u(\cdot,k)=0$ that satisfy the boundary conditions along with certain integrability conditions. Considering each edge separately, one recalls that a fundamental system of the equation 
\begin{align}\label{ew+}
-u^{\prime\prime}(x,k)= k^2 u(x,k) &&  \mbox{with}\quad k\neq 0,
\end{align}
$u\colon \R \rightarrow \C$ is given by $e^{ikx}$ and $e^{-ikx}$ and consequently of 
\begin{align}\label{ew-}
u^{\prime\prime}(x,\kappa)= \kappa^2 u(x,\kappa), &&  \mbox{for}\quad \kappa \neq 0,
\end{align}
by $e^{-\kappa x}$ and $e^{\kappa x}$. As the two fundamental systems have different integrability properties in certain regions of $\C$, one defines the two open quadrants
\begin{align*}
\Qe=& \{k\in \C \mid \Re(k)>0 \ \mbox{and} \ \Im(k)>0  \}
\end{align*}
and 
\begin{align*}
\Pe=& \{k\in \C \mid \Re(k)<0 \ \mbox{and} \ \Im(k)>0  \}.
\end{align*} 
For $k,\kappa\in \Qe$ the functions $e^{ikx}$ and $e^{-\kappa x}$ are elements of $L^2([0,\infty))$, whereas $e^{-ikx}$ and $e^{\kappa x}$ are not in $L^2([0,\infty))$. For $k,\kappa,\in \Pe$ the functions $e^{ikx}$ and $e^{\kappa x}$ are square integrable. 

\subsection{Resonances and non--zero eigenvalues}
A general \textit{Ansatz} for a solution that satisfies simultaneously equation \eqref{ew+} on the positive edges and equation \eqref{ew-} on the negative edges is the function $\psi$ defined by
\begin{align}\label{Ansatz1IndQG}
\psi(x,k,i\kappa)= \begin{cases}
s_j(k)e^{ikx_j}, & j\in \Ee_+, \\
          \alpha_{j}(k) e^{ik x_j} + \beta_{j}(k) e^{-ik x_j}, & j\in \Ie_+,\\
           s_j(i\kappa) e^{-\kappa x_j}, & j\in \Ee_-, \\
           \alpha_{j}(i\kappa) e^{-\kappa x_j} + \beta_{j}(i\kappa) e^{\kappa x_j}, & j\in \Ie_-
                    \end{cases} 
\end{align}
with $k,\kappa\in \C\setminus\{0\}$. With this notation one has 
\begin{eqnarray*}
(\tau-k^2)\psi(x,k,\pm ik)=0  & \mbox{and} & (-\Delta-k^2)\psi(x,k,\pm k)=0.
\end{eqnarray*}
Assuming that $\Ee_-\neq \emptyset$ and $\Ee_+\neq \emptyset$ one has for $k\in \Qe$, that $\psi(\cdot,k,ik)$ is square integrable. Analogue for $k \in \Pe$, $\psi(\cdot,k,-ik)$ is square integrable on the external edges, as now $e^{kx}\in L^2([0,\infty))$. Introduce for brevity the notations
\begin{eqnarray*}
 \chi_+(k)= \begin{bmatrix} \{s_j(k)\}_{j\in \Ee_+} \\ \{\alpha_j(k)\}_{j\in \Ie_+} \\ \{\beta_j(k)\}_{j\in \Ie_+} \end{bmatrix}, &
 \chi_-(i\kappa)= \begin{bmatrix} \{s_j(i\kappa)\}_{j\in \Ee_-} \\ \{\alpha_j(i\kappa)\}_{j\in \Ie_-} \\ \{\beta_j(i\kappa)\}_{j\in \Ie_-} \end{bmatrix}, & 
 \chi(k,i\kappa)=\begin{bmatrix} \chi_+(k) \\ \chi_-(i\kappa)  \end{bmatrix}.  
\end{eqnarray*}
For the boundary values of the \textit{Ansatz} function $\psi(\cdot,k,i\kappa)$ one obtains the formulae
\begin{align*}
  \underline{\psi(\cdot,k,i\kappa)}= X_{n,m}(k,i\kappa,\au) \chi(k,ik) , && 
  \underline{\psi(\cdot,k,i\kappa)}^{\prime}=Y_{n,m}(k,i\kappa,\au)\chi(k,ik)   
\end{align*}
with
\begin{align*}\index{$X_{n,m}(k,i\kappa,\au)$}\index{$Y_{n,m}(k,i\kappa,\au)$}
X_{n,m}(k,i\kappa,\au)&= \begin{bmatrix} X_{n}(k,\au_+) & 0 \\ 0 &  X_{m}(i\kappa,\au_-) \end{bmatrix} \ \mbox{and}
\\ Y_{n,m}(k,i\kappa,\au) &=\begin{bmatrix}  Y_{n}(k,\au_+) & 0 \\ 0 &   Y_{m}(i\kappa,\au_-) \end{bmatrix}
\end{align*}
matrices in $\Ke$, which are written with respect to the decomposition $\Ke=\Ke_+\oplus \Ke_- $. This uses the notation  
\begin{eqnarray*}
X_{l}(k,\bu)= \begin{bmatrix} \mathds{1}& 0 & 0 \\ 0 & \mathds{1} & \mathds{1} \\ 0 & e^{ik\bu} & e^{-ik\bu} \end{bmatrix} &\mbox{and}&
Y_{l}(k,\bu)= ik\begin{bmatrix} \mathds{1} & 0 & 0 \\ 0 & \mathds{1} & -\mathds{1}\\ 0 & e^{ik\bu} & -e^{ik\bu} \end{bmatrix},
\end{eqnarray*} 
where for $l=n$ one has to plug in $\bu=\au_+$ and for $l=m$ one inserts $\bu=\au_-$. Actually these are the matrices known from the spectral theory of Laplace operators on $(\Ge_+,\au_+)$ and $(\Ge_-,\au_-)$, see for example \cite[Section 3]{VKRS2006}. The matrices $X_{n}(k,\au_+)$ and $Y_{n}(k,\au_+)$ define operators in $\Ke_+$ and $X_{m}(k,\au_-)$ and $Y_{m}(k,\au_-)$ in $\Ke_-$, respectively. The matrix $e^{ik\bu}$ denotes the diagonal matrix with entries $\{e^{ik\bu}\}_{k,l}= \delta_{k,l} e^{ik b_l}$, where $b_l\in \bu$ and $\delta_{k,l}$ denotes the Kronecker delta. 

Note that there is a function $\psi(\cdot,k,i\kappa)$ of the form \eqref{Ansatz1IndQG} satisfying the boundary conditions $A \underline{\psi(\cdot,k,i\kappa)}+B \underline{\psi(\cdot,k,i\kappa)}^{\prime}=0$, if and only if there exist coefficients $\chi(k, i\kappa)\in \Ke\setminus\{0\}$ such that
\begin{align}\label{EWIndQG}\index{$Z_{n,m}(A,B,k,i\kappa,\underline{a})$}
Z_{n,m}(A,B,k,i\kappa,\underline{a})\chi(k,i\kappa) =0 
\end{align}
holds with 
$$Z_{n,m}(A,B,k,i\kappa,\underline{a})= A X_{n,m}(k,i\kappa,\au) + B Y_{n,m}(k,i\kappa,\au).$$ 
Since $\Ke$ is a finite dimensional space the condition
\begin{align*}
\det Z_{n,m}(A,B,k,i\kappa,\underline{a})=0
\end{align*}
is equivalent to the condition given in equation \eqref{EWIndQG}. 
\begin{remark}\label{SingpointofZ}
Let $T(A,B)$ be self-adjoint. Then there are no functions $\psi(\cdot,k, ik)$ of the form \eqref{Ansatz1IndQG}, for $k\in \Qe$ with $(\tau-k^2)\psi(\cdot,k, ik)=0$ and $\psi(\cdot,k, ik)\in \Dom(T(A,B))$. Similarly for $k\in \Pe$, there is no such $\psi(\cdot,k, -ik)$ with $(\tau-k^2)\psi(\cdot,k, -ik)=0$ that satisfies $\psi(\cdot,k, -ik)\in \Dom(T(A,B))$. This is due to the fact that $\psi(\cdot,k, ik)$ and $\psi(\cdot,k, -ik)$, respectively would be eigenfunctions to the complex eigenvalue $k^2\in \C\setminus\R$, which would contradict the self--adjointness of $T(A,B)$. The question where the resolvent $R(k)=(T(A,B)-k^2)^{-1}$ as function in $k$ admits a meromorphic continuation to the real line is related to equation \eqref{EWIndQG}. The poles of the meromorphic continuation of $R(\cdot)$ are linked to certain singular points of $Z_{n,m}(A,B,\cdot,\cdot,\au)$. 
\end{remark}
For star graphs, that is for graphs with $\Ie=\emptyset$, the matrix $Z_{n,m}(A,B,k,i\kappa,\au)$ simplifies to
\begin{align*}
Z_{n,m}(A,B,k,i\kappa) = A + B \begin{bmatrix} ik & 0 \\ 0 & -\kappa  \end{bmatrix}.
\end{align*}
Instead of the matrices $A$ and $B$ one can consider the equivalent parametrization according to Corollary~\ref{PL}. There is an orthogonal projector $P$ and a Hermitian matrix $L$ acting in $\Ker P$ such that with $A^{\prime}= L+P$ and $B^{\prime}= P^{\perp}J_{n,m}$ one has $T(A^{\prime},B^{\prime})=T(A,B)$. Denote by $l_j$ the eigenvalues of $L$ and by $P_j$ the orthogonal projector on the corresponding eigenspace. Then one has for $k,\kappa\in \C\setminus\{0\}$ the representation
\begin{align*}
\displaystyle{Z_{n,m}(A,B,k,i\kappa) = P+ \sum_{j} P_j \left(l_j  +  \begin{bmatrix} -ik & 0 \\ 0 & \kappa \end{bmatrix}\right)}.   
\end{align*}

Analogue to \cite[Theorem 3.1]{VKRS1999} one proves now
\begin{lemma}\label{zeroZ}
Let $T(A,B)$ be self--adjoint. Then zeros of the function 
$$k \mapsto \det Z_{n,m}(A,B,k,ik,\underline{a})$$ 
are discrete and there are no zeros in $\Qe$. For $\Ie=\emptyset$ there are $\abs{\Ee}$ zeros in $\overline{\Qe}$ at the most. The analogue statements hold for the zeros of $\kappa \mapsto\det Z_{n,m}(A,B,i\kappa,\kappa,\underline{a})$.
\end{lemma}  
\begin{proof}
The function $k \mapsto \det Z_{n,m}(A,B,k,ik,\underline{a})$ is an entire function. This follows directly from the definition. Therefore it can be represented as a power series. If the zeros had not been  discrete, then there would be a convergent sequence $\{k_j\}_{j\in \mathbb{N}}$ such that $\det Z_{n,m}(A,B,k_j,ik_j, \underline{a})=0$. This would imply that $\det Z_{n,m}(A,B,k,ik,\underline{a}) \equiv 0$. Therefore it is sufficient to find one point $k_0\in \C$ for which $\det Z_{n,m}(A,B,k_0,ik_0,\underline{a}) \neq 0$. If there had been such a $k_0$ with $\det Z_{n,m}(A,B,k_0,ik_0,\underline{a})=0$ from the open quadrant $\Qe$, this would imply the integrability of the \textit{Ansatz} function $\psi(x,k_0,ik_0)$. According to Remark~\ref{SingpointofZ} then $k_0^2\in \C\setminus\R$ would be an eigenvalue of $T(A,B)$, which is a contradiction to the self--adjointness of $T(A,B)$. 
\end{proof}

Putting the pieces together one obtains
\begin{proposition}
Let $T(A,B)$ be self-adjoint. Then $k^2>0$ is an eigenvalue of $T(A,B)$ if and only if there is a coefficient $\chi(k,ik)\neq 0$ with $s_j(k)=0$ for $j\in \Ee_+$ and for the positive square roots $k>0$ of $k^2$ such that
\begin{eqnarray*}
Z_{n,m}(A,B,k,ik,\underline{a})\chi(k,ik) =0. 
\end{eqnarray*}
The number $-\kappa^2<0$ is a negative eigenvalue of $T(A,B)$ if and only if for the positive square root $\kappa>0$ of $\kappa^2$ there is a coefficient $\chi(i\kappa,\kappa)\neq 0$ with $s_j(\kappa)=0$ for $j\in \Ee_-$ such that
\begin{eqnarray*}
Z_{n,m}(A,B,i\kappa,\kappa,\underline{a})\chi(i\kappa,\kappa) =0 &\mbox{with} & s_j(\kappa)=0. 
\end{eqnarray*}
\end{proposition}
In particular for $\Ee=\emptyset$ one obtains 
\begin{corollary}
Let $T(A,B)$ be self-adjoint and assume that $\Ee=\emptyset$. Then the positive square roots $k>0$ of positive eigenvalues $k^2$ are exactly the solutions of the secular equation  
\begin{eqnarray*}
\det Z_{n,m}(A,B,k,ik,\underline{a})=0
\end{eqnarray*}
and the positive square roots $\kappa>0$ of the absolute values of negative eigenvalues $-\kappa^2$ are exactly the solutions of the secular equation
\begin{eqnarray*}
\det Z_{n,m}(A,B,i\kappa, \kappa,\underline{a})=0. 
\end{eqnarray*}
\end{corollary}
\begin{remark}
In the presence of external edges only the solutions that are square integrable on the external edges are eigenvalues. In general one cannot exclude other resonances in the set of singular points of the function $k \mapsto Z_{n,m}(A,B,k,ik)$. In all the examples studied the resonances are only eigenvalues, but the author has not been able to find a general proof. 
\end{remark}
  
\subsection{Eigenvalue zero.} The solutions of $u^{\prime\prime}(x)=0$ are the affine functions and only the trivial solution is square integrable on $[0,\infty)$. Thus the eigenfunctions of $T(A,B)$ to the eigenvalue zero have to be affine on each edge and one gets to the \textit{Ansatz} 
\begin{align*}
\psi^0(x)= \begin{cases}
           0, & j\in \Ee_+\cup \Ee_-, \\
           \alpha_j^0  + \beta_j^0 x_j, & j\in  \Ie_+ \cup \Ie_-.\\
        \end{cases}
\end{align*}
The coefficients appearing on the right hand side are written into the vectors
\begin{eqnarray}\label{EW0} 
 \chi_+^0= \begin{bmatrix} 0 \\ \{\alpha_j^0\}_{j\in \Ie_+} \\ \{\beta_j^0\}_{j\in \Ie_+} \end{bmatrix} &\mbox{and} & 
 \chi_-^0= \begin{bmatrix} 0 \\ \{\alpha_j^0\}_{j\in \Ie_-} \\ \{\beta_j^0\}_{j\in \Ie_-} \end{bmatrix}, 
\end{eqnarray}
which are summarized in one vector
\begin{eqnarray*}
 \chi^0=\begin{bmatrix} \chi_+^0 \\ \chi_-^0  \end{bmatrix}.  
\end{eqnarray*}
For the boundary values one obtains 
\begin{eqnarray*}
   \underline{\psi^0}= X_{n,m}^0(\au) \chi^0  & \mbox{and} & 
   \underline{\psi^0}^{\prime}= Y_{n,m}^0(\au) \chi^0
\end{eqnarray*}
with
\begin{eqnarray*}
X_{n,m}^0(\au)=\begin{bmatrix}X_n^0(\au_+) & 0 \\ 0 & X_m^0(\au_-)\end{bmatrix} & \mbox{and} & Y_{n,m}^0(\au)=\begin{bmatrix}Y_n^0(\au_+) & 0 \\ 0 & Y_m^0(\au_-)\end{bmatrix},
\end{eqnarray*}
which are matrices in $\Ke$, that are written with respect to the decomposition  $\Ke=\Ke_+\oplus \Ke_- $. This uses the notation 
\begin{eqnarray*} 
X^0_{l}(\bu)=\begin{bmatrix}0 & 0 &0 \\ 0 & \mathds{1} &0 \\ 0 & \mathds{1} & \bu  \end{bmatrix} & \mbox{and}& 
Y^0_{l}(\bu)=\begin{bmatrix}0 & 0 &0 \\ 0 & 0 & \mathds{1}\\ 0 & 0 & -\mathds{1}  \end{bmatrix},
\end{eqnarray*}
where for $l=n$ one has to plug in $\bu=\au_+$ and for $l=m$ one inserts $\bu=\au_-$. Actually these are the matrices known from the spectral theory of Laplace operators on finite metric graphs. The matrices $X_{n}^0(\au_+)$ and $Y_{n}^0(\au_+)$ define operators in $\Ke_+$, and the matrices $X_{m}^0(\au_-)$ and $Y^0_{m}(\au_-)$ define operators in $\Ke_-$. The matrix $\bu$ denotes the diagonal matrix with entries $\{\bu\}_{k,l}= \delta_{k,l} b_l$. This gives
\begin{proposition}
Let $\Ee=\emptyset$. Then $T(A,B)$ is invertible if and only if
$$\det (AX_{n,m}^0(\au)+BY_{n,m}^0(\au))=0.$$ 
The dimension of $\Ker T(A,B)$ is equal to the dimension of the space \\ $\Ker(AX_{n,m}^0(\au)+BY_{n,m}^0(\au))$. In the presence of external edges, zero is an eigenvalue of $T(A,B)$ if and only if there exists a vector $\chi^0$ with entries of the form \eqref{EW0} such that 
$$(AX_{n,m}^0(\au)+BY_{n,m}^0(\au))\chi^0=0.$$
\end{proposition}

\begin{remark}
Since the solutions of $u^{\prime\prime}(x)=0$ are affine functions one can use the same \textit{Ansatz} $\psi^0$ to find solutions of the Laplace equation $\Delta u=0$ on the graph. Consequently zero is an eigenvalue of $T(A,B)$ if and only if zero is an eigenvalue of the not necessarily self-adjoint Laplace operator $-\Delta(A,B)$. Note that $-\Delta(A,B)$ has the same domain as $T(A,B)$, since it is defined by the same boundary conditions. 
\end{remark}

\subsection{Eigenvalue asymptotics.} Assume $T(A,B)$ to be self--adjoint, then its eigenvalues, according to Lemma~\ref{zeroZ}, are discrete with only finite multiplicities. Denote by $0<\lambda_1^+ \leq \lambda_2^+ \leq \ldots$ the positive eigenvalues, counted with multiplicities and by $0<\lambda_1^- \leq \lambda_2^- \leq \ldots$ the absolute values of the negative eigenvalues, counted with multiplicities. Note that in general the eigenvalues of $T(A,B)$ can be embedded in the continuous spectrum. One defines the counting function for the positive eigenvalues $N_{+}(\lambda; T(A,B))$ and the counting function for the negative eigenvalues $N_{-}(\lambda; T(A,B))$ by
$$N_{\pm}(\lambda; T(A,B)) = \sum_{\lambda_{j}^{\pm}<\lambda} 1.$$
In certain cases a generalized Weyl law holds for the asymptotic behaviour of the eigenvalue counting functions. 

\begin{proposition}\label{Weylgraph}  \ \ \ \ \ \ \  \ \ \ \  \ \  \
\begin{enumerate}
\item Assume that $(\Ge,\au)$ is a compact metric graph, that is $\Ee=\emptyset$, and let $T(A,B)$ be self-adjoint. Then asymptotically the eigenvalue counting functions are  
$$\displaystyle{N_{\pm}(\lambda; T(A,B))\sim \frac{1}{\pi}\bigg(\sum_{i\in\mathcal{I}_{\pm}} a_i\bigg)\, \lambda^{1/2}, \ \mbox{for} \ \lambda\to\infty.}$$
\item Assume that $(\Ge,\au)$ is non--compact, but $\Ee_{-}=\emptyset\,(\Ee_+=\emptyset)$ and let $T(A,B)$ be self-adjoint. Then the asymptotic formula for $N_{-}(\lambda; T(A,B))$ \\ $(N_{+}(\lambda; T(A,B)))$ still holds.
\end{enumerate}
\end{proposition}
\begin{proof}
Suppose that $\Ee=\emptyset$. Then it follows from the compact embedding of $\De \hookrightarrow \He$ that $T(A,B)$ has compact resolvent and therefore there exists $k^2\in \R$ with $k^2\in \rho(T(A,B))\cap \rho(T(\mathds{1},0))$. According to the forthcoming Proposition~\ref{resolvent}, the difference $$(T(A,B)-k^2)^{-1}-(T(\mathds{1},0)-k^2)^{-1}$$ is a self-adjoint finite rank operator. To prove the first statement, one applies the $\min-\max$--principle to the resolvents. As the difference has only finite rank one obtains that 
$$N_{\pm}(\lambda; T(A,B)) \sim N_{\pm}(\lambda; T(\mathds{1},0)), \ \mbox{for} \ \lambda\to\infty.$$
The operator $T(\mathds{1},0)$ decouples all edges and defines (positive) Dirichlet Laplacians on each positive edge and Dirichlet Laplacians multiplied with minus one on each negative edge. The claim follows by summing up the eigenvalue counting functions of the Dirichlet Laplacians on the intervals.  

To prove the second statement one takes into account that for a bounded operator one can apply the $\min-\max$--principle to characterize the discrete spectrum below the essential spectrum. For the case of $\Ee_+\neq \emptyset$ and $\Ee_-=\emptyset$ the essential spectrum of $T(A,B)$ is $[0,\infty)$ and the negative spectrum is discrete. Denote by $-\mu^-_1<0$ the largest negative eigenvalue of $T(\mathds{1},0)$. Then all $-k^2\in \left(\max\{-\lambda^-_1,-\mu_1^-\}, 0\right)$ are in the resolvent set of both $T(A,B)$ and $T(\mathds{1},0)$. For such a $-k^2$ the difference $$(T(A,B)+k^2)^{-1}-(T(\mathds{1},0)+k^2)^{-1}$$ is a self-adjoint finite rank operator and the negative spectrum of $(T(A,B)+k^2)^{-1}$ consists of a sequence of eigenvalues of finite multiplicities that accumulate at zero. Therefore one can apply the $\min-\max$--principle to characterize the positive eigenvalues of $-(T(A,B)+k^2)^{-1}$ as well as the positive eigenvalues of $-(T(\mathds{1},0)+k^2)^{-1}$. As one has again only a perturbation of finite rank the claim follows by the same reasoning as above.    
\end{proof}
For the proof of Proposition~\ref{Weylgraph} it is essential to consider only finite rank perturbations.

\begin{remark}
Let $(\Ge,\au)$ be a compact finite metric graph and assume that $T(A,B)$ is self--adjoint and invertible. Then the operator $T(A,B)$ has a compact inverse and the eigenvalues of $T(A,B)^{-1}$ can be determined as the successive extrema of the Rayleigh-quotient 
\begin{eqnarray*}
\rho[\psi]:= \frac{\langle \tau\psi, \psi \rangle}{\langle \tau\psi, \tau \psi \rangle}, & \psi \in \Dom(T(A,B)).
\end{eqnarray*}
This delivers a variational characterizations of the eigenvalues of $T(A,B)$.
\end{remark}

\subsection{Generalized eigenfunctions.}
An important tool for the study of the absolutely continuous part of self--adjoint operators $T(A,B)$ and hence for their scattering theory are generalized eigenfunctions. Assume that $\Ee\neq \emptyset$. For $l\in \Ee_+$ one \textit{Ansatz} for a generalized eigenfunction is
\begin{align*}\index{$\varphi_{l}(x,k,i\kappa)$}
\varphi_{l}(x,k,i\kappa)= \begin{cases}
\delta_{lj}e^{-ikx_l} + s_{lj}(k)e^{ikx_l}, & j\in \Ee_+, \\
          \alpha_{lj}(k) e^{ik x_j} + \beta_{lj}(k) e^{-ik x_j}, & j\in \Ie_+,\\
           s_{lj}(i\kappa ) e^{-\kappa x_j}, & j\in \Ee_-, \\
           \alpha_{lj}(i\kappa) e^{-\kappa x_j} + \beta_{lj}(i\kappa) e^{\kappa x_j}, & j\in \Ie_-
                    \end{cases} 
\end{align*}
and for $l\in \Ee_-$
\begin{align*}
\varphi_{l}(x,k,i\kappa)= \begin{cases}
s_{lj}(k)e^{ikx_l}, & j\in \Ee_+, \\
          \alpha_{lj}(k) e^{ik x_j} + \beta_{lj}(k) e^{-ik x_j}, & j\in \Ie_+,\\
      \delta_{lj} e^{\kappa x_j}  +  s_{lj}(i\kappa ) e^{-\kappa x_j}, & j\in \Ee_-, \\
           \alpha_{lj}(i\kappa) e^{-\kappa x_j} + \beta_{lj}(i\kappa) e^{\kappa x_j}, & j\in \Ie_-
                    \end{cases} 
\end{align*}
with $\delta_{lj}$ the Kronecker delta and sought after coefficients $s_{lj}(k), \alpha_{lj}(k)$ and $\beta_{lj}(k)$. The functions $\varphi_{l}(\cdot,k,i\kappa)$, $l\in \Ee$, are solutions of equation \eqref{ew+} on the positive edges and of equation \eqref{ew-} on the negative edges, but they are not square integrable. For notational simplicity these \textit{Ansatz} functions are written into the  $\abs{\Ee}\times \abs{\Ee\cup \Ie}$--matrix  $\varphi(x,k,i\kappa)$ with entries
$$\left\{\varphi(x,k,i\kappa)\right\}_{lj}:=\left\{\varphi_{l}(x,k,i\kappa)\right\}_j,$$
that is the $(l,j)$-th entry of $\varphi(\cdot,k,i\kappa)$ is the restriction of $\varphi_l$, $l\in \Ee$ to $I_j$, $j\in \Ee\cup \Ie$. The coefficients for $l\in \Ee_+$ are written into the vectors
\begin{eqnarray*}\index{$ \chi_{l,\pm}(k)$}
 \chi_{l,+}(k)= \begin{bmatrix} \{ \delta_{lj} + s_{lj}(k)\}_{j\in \Ee_+} \\ \{\alpha_{lj}(k)\}_{j\in \Ie_+} \\ \{\beta_{lj}(k)\}_{j\in \Ie_+} \end{bmatrix}, & &
 \chi_{l,-}(i\kappa)= \begin{bmatrix} \{s_{lj}(i\kappa)\}_{j\in \Ee_-} \\ \{\alpha_{lj}(i\kappa)\}_{j\in \Ie_-} \\ \{\beta_{lj}(i\kappa)\}_{j\in \Ie_-} \end{bmatrix}
\end{eqnarray*}
and for $l\in \Ee_-$ into
\begin{eqnarray*}
 \chi_{l,+}(k)= \begin{bmatrix} \{s_{lj}(k)\}_{j\in \Ee_+} \\ \{\alpha_j(k)\}_{j\in \Ie_+} \\ \{\beta_j(k)\}_{j\in \Ie_+} \end{bmatrix}, & & 
 \chi_{l,-}(i\kappa)= \begin{bmatrix} \{ \delta_{lj} +  s_j(i\kappa)\}_{j\in \Ee_-} \\ \{\alpha_j(i\kappa)\}_{j\in \Ie_-} \\ \{\beta_j(i\kappa)\}_{j\in \Ie_-} \end{bmatrix}.
\end{eqnarray*}
Finally one sets 
$$ \chi_{l}(k,i\kappa)=\begin{bmatrix} \chi_{l,+}(k) \\ \chi_{l,-}(i\kappa)  \end{bmatrix}, $$
which is a vector in $\Ke=\Ke_+\oplus\Ke_-$. These $\chi_{l}$ are summarized into the $\abs{\Ee}\times \abs{\Ee\cup \Ie}$ matrix $\chi(k,i\kappa)$ with entries
$$\left\{\chi(k,i\kappa)\right\}_{lj}:=\left\{\chi_{l}(k,i\kappa)\right\}_j.$$  \index{$\chi(k,i\kappa)$}
For the boundary values of the \textit{Ansatz} functions $\varphi_l(\cdot,k,i\kappa)$ defined above one obtains the formulae
\begin{align*}
  \underline{\varphi(\cdot,k,i\kappa)}&= \mathds{1}_{n+m} e_{n,m}  + X_{n,m}(k,i\kappa,\au) \chi(k,i\kappa) , \\
  \underline{\varphi(\cdot,k,i\kappa)}^{\prime}&=I_{n,m}(-k,-i\kappa)e_{n,m} + Y_{n,m}(k,i\kappa,\au)\chi(k,i\kappa),
\end{align*}
where for brevity
\begin{eqnarray*}\index{$I_{n,m}(k,i\kappa)$, $I_{n,m}$}
I_{n,m}(k,i\kappa)=\begin{bmatrix} ik\mathds{1}_{n} & 0 \\ 0 & -\kappa\mathds{1}_{m}\end{bmatrix}, &  & I_{n,m}= \begin{bmatrix} \mathds{1}_{n} & 0 \\ 0 & i\mathds{1}_{m}\end{bmatrix},
\end{eqnarray*}
which are related by $ik\cdot I_{n,m}=I_{n,m}(k,ik)$. The notation
\begin{eqnarray*}\index{$e_{n,m}$}
e_{n,m}= \begin{bmatrix} e_n & 0 \\ 0 & e_m  \end{bmatrix}, &\mbox{where} & e_{p}= \begin{bmatrix} \mathds{1} \\ 0 \\  0   \end{bmatrix}, \quad p\in \{n,m\}
\end{eqnarray*}
is used, where $e_{n}$ is a $\abs{\Ke_{\Ee_+}}\times \abs{\Ke}$-matrix, $e_{m}$ a $\abs{\Ke_{\Ee_-}}\times \abs{\Ke}$-matrix and consequently the block-operator matrix $e_{n,m}$ is a $\left(\abs{\Ke_{\Ee_+}}+\abs{\Ke_{\Ee_-}}\right)\times \abs{\Ke}$-matrix. The boundary conditions
\begin{eqnarray*}
 A \underline{\varphi_l(\cdot,k,i\kappa)}+ B\underline{\varphi_l(\cdot,k,i\kappa)}^{\prime}=0, &\mbox{for }  l\in \Ee
\end{eqnarray*}
are satisfied for all $\varphi_l(\cdot, k, i\kappa)$, $l\in \Ee$ if and only if  
\begin{align*}
 ( AX_{n,m}(k,i\kappa,\au) + BY_{n,m}(k,i\kappa,\au))\chi(k,i\kappa)=-(A \mathds{1}_{n+m}+  B I_{n,m}(-k,-i\kappa))e_{n,m}
\end{align*}
holds for the coefficient vector $\chi(k,i\kappa)$. If $AX_{n,m}(k,i\kappa,\au) + BY_{n,m}(k,i\kappa,\au)$ is invertible one can define the matrix valued transform 
\begin{align*}
 \mathfrak{X}(k,i\kappa)=-( AX_{n,m}(k,i\kappa,\au) + BY_{n,m}(k,i\kappa,\au))^{-1}(A \mathds{1}_{n+m}+  B I_{n,m}(-k,-i\kappa)).
\end{align*}
This is  a transform in $A$ and $B$ as well as in $k,\kappa$. However in the notation the dependence on $A$ and $B$ is omitted. When $\mathfrak{X}(k,i\kappa)$ is well-defined one can solve the 
equation for the coefficients $\chi(k,i\kappa)$, which are then uniquely determined by 
\begin{align}\label{chi}\index{$\mathfrak{X}(k,i\kappa)$}
 \chi(k,i\kappa)=  \mathfrak{X}(k,i\kappa)  e_{n,m}.
\end{align}
Note that this way one obtains that $\chi(\pm k,ik)$ is well-defined for all $k>0$, except for resonances with $\det Z_{n,m}(A,B, \pm k,ik,\au)=0$. With the coefficients $\chi(\pm k,ik)$, where $k>0$ such that $Z_{n,m}(A,B, \pm k,ik,\au)$ is regular, the functions 
\begin{eqnarray*}
\varphi_l(\cdot,k,ik) & \mbox{and} & \varphi_l(\cdot,-k,ik), \quad l\in \Ee_+, 
\end{eqnarray*}
satisfy the boundary conditions $A\underline{\varphi} + B \underline{\varphi_l}^{\prime}=0$. By construction they additionally solve the equation $\left( \tau -k^2 \right)\varphi =0$. Therefore they are denoted as \textit{generalized eigenfunctions of $T(A,B)$ for positive energies}. Analogue for $l\in \Ee_-$ and $\kappa>0$ such that $Z_{n,m}(A,B, i\kappa,\pm\kappa,\au)$ is regular the functions    
\begin{eqnarray*}
\varphi_l(\cdot,i\kappa,\kappa)& \mbox{and} & \varphi_l(\cdot,i\kappa,-\kappa)
\end{eqnarray*}
are denoted as \textit{generalized eigenfunctions of $T(A,B)$ for negative energies}. These statements are going to be specified, when considering scattering problems related to the operator $T(A,B)$ in Section~\ref{secscat}. There the associated wave operators are computed in terms of the generalized eigenfunctions of $T(A,B)$.

\subsection{Coefficient matrix.} 
For a star graph, that is for $\Ie=\emptyset$, the coefficients of the generalized eigenfunctions are simplified to become
\begin{align*}
 \chi(k,i\kappa)=-(A + B I_{n,m}(k,i\kappa) )^{-1}(A + B I_{n,m}(-k,-i\kappa)).
\end{align*}
Given an arbitrary finite metric graph $(\Ge,\au)$ and self-adjoint boundary conditions parametrized by matrices $A$ and $B$ one can cut all internal edges in twain and substitute each of the pieces by an external edge. Then the graph becomes a union of star graphs, on which $A$ and $B$ still define a self-adjoint operator. This motivates the definition of the \textit{local coefficient matrix of the generalized eigenfunctions} $\Cm(k,i\kappa)$ for any graph and fixed $A,B$ satisfying the assumptions of Theorem~\ref{sabc} by 
\begin{align}\label{Cm}\index{$ \Cm(k,i\kappa)$, $\Cm_{\pm\pm}(k,i\kappa)$}
 \Cm(k,i\kappa):=-(A + B I_{n,m}(k,i\kappa) )^{-1}(A - B I_{n,m}(k,i\kappa)).
\end{align}
In short $\Cm$ is denoted the \textit{coefficient matrix}. For star graphs $\Cm(k,i\kappa)=\mathfrak{X}(k,i\kappa)$ holds. The coefficient matrix is an analogue of the local scattering matrix that appears in the context of positive Laplacians on metric graphs, see for example \cite{VKRS1999}, and in the case $\Ee_-=\emptyset$ both objects coincide. With respect to the block decomposition induced by $\mathcal{E}_+\cup\mathcal{I}_+ $ and  $\mathcal{E}_-\cup \mathcal{I}_-$ the coefficient matrix can be written as a block-operator matrix 
\begin{align*}
\Cm(k,i\kappa)=\begin{bmatrix}
\Cm_{++}(k,i\kappa) & \Cm_{+-}(k,i\kappa) \\
\Cm_{-+}(k,i\kappa) & \Cm_{--}(k,i\kappa)
\end{bmatrix},
\end{align*}
where $\Cm_{\pm \pm}(k,i\kappa)$ act in $\Ke_{\pm}$ and $\Cm_{\mp \pm}(k,i\kappa)$ are maps of $\Ke_{\mp}$ to $\Ke_{\pm}$. The poles of $\Cm(k,\pm ik)$, counted with multiplicities, are only finitely many, because \\ $\det (A + B I_{n,m}(k,\pm ik) )$ is a polynomial of degree not greater than $n+m$ and therefore it can have $n+m$ zeros at the most. 

\begin{example}\label{exCm}
Consider the operator $T(A,B)$ from Example~\ref{ex1IndQG}. Then one obtains the local coefficient matrix 
\begin{align*}
\Cm(k,i\kappa)=\frac{1}{\kappa^2 +k^2}\begin{bmatrix} k^2-\kappa^2 + 2i \kappa k & 2\kappa^2 - 2i \kappa k  \\  2 k^2 + 2i \kappa k  & k^2-\kappa^2 - 2i \kappa k \end{bmatrix}.
\end{align*}
The local coefficient matrix $\Cm=\Cm(k,ik)$ is the $k$--independent matrix 
\begin{align*}
\Cm=\begin{bmatrix} i & (1-i) \\  (1+i) & -i \end{bmatrix}.
\end{align*}
This gives the generalized eigenfunction for positive energy $k>0$
\begin{align*}
\varphi_1(x,k,ik)= \begin{cases} e^{-ikx} + i e^{ikx}, & x \in \Ee_+, \\ (1+i) e^{-kx}, & x \in \Ee_-.\end{cases}
\end{align*}
\end{example}
Looking at the generalized eigenfunction for positive energy $\varphi_1(\cdot,k,ik)$ one can interpret the off-diagonal entry $\Cm_{+-}$ of $\Cm$  as depth of penetration of the incoming wave propagating from the positive edge into the negative part of the graph. The diagonal entry $\Cm_{++}$ describes the scattering behaviour of the incoming waves, in this example it is the reflection coefficient. The entries $\Cm_{-+}$ and $\Cm_{--}$ have corresponding interpretations for negative energies.

There are some useful symmetries of $\Cm(k,i\kappa)$ summarized in the following
\begin{lemma}\label{cmsym}
For complex $k,\kappa\neq 0$ and $A,B$ satisfying the assumptions of Theorem~\ref{sabc} one has
\begin{enumerate}
\item $\Cm(k,i\kappa)^{-1}=\Cm(-k,-i\kappa)$ and 
\item $\Cm(k,ik) \begin{bmatrix} \mathds{1} & 0 \\ 0 & i\end{bmatrix} = \begin{bmatrix} \mathds{1} & 0 \\ 0 & i\end{bmatrix}\Cm(-\overline{k},i\overline{k})^*$. 
\item Let be $\hat{A}=P$ and $\hat{B}=P^{\perp}J_{n,m}$, where $P$ is an orthogonal projector in $\Ke$. Then the matrices $\Cm(k,ik)$ and $\Cm(k,-ik)$ are $k$--independent.
\end{enumerate}
\end{lemma}

\begin{proof}
(i) The first statement follows immediately from the definition since \\ $I_{n,m}(-k,-i\kappa)= -I_{n,m}(k,i\kappa)$ holds. 

(ii) To prove the second statement, consider instead of $A$ and $B$ the equivalent parametrization given in Corollary~\ref{PL} with 
\begin{eqnarray*}
\hat{A}=L +P & \mbox{and} & \hat{B}=P^{\perp}\begin{bmatrix} \mathds{1} & 0 \\ 0 & -\mathds{1}\end{bmatrix},
\end{eqnarray*}
where $P$ is an orthogonal projector and $L$ is a Hermitian operator with $L=P^{\perp}LP^{\perp}$. This gives
\begin{align*}
  &\left( \hat{A}-ik P^{\perp} \begin{bmatrix} \mathds{1} & 0 \\ 0 & -i\end{bmatrix} \right) \begin{bmatrix} \mathds{1} & 0 \\ 0 & i\end{bmatrix} \left( \hat{A}+ ik  \begin{bmatrix} \mathds{1} & 0 \\ 0 & -i\end{bmatrix}P^{\perp} \right) \\
 =&\left( \hat{A}+ ik P^{\perp} \begin{bmatrix} \mathds{1} & 0 \\ 0 & -i\end{bmatrix} \right) \begin{bmatrix} \mathds{1} & 0 \\ 0 & i\end{bmatrix} \left( \hat{A}- ik  \begin{bmatrix} \mathds{1} & 0 \\ 0 & -i\end{bmatrix}P^{\perp} \right).
\end{align*}
For
$$\Cm(k,ik)= \left( \hat{A}+ik P^{\perp} \begin{bmatrix} \mathds{1} & 0 \\ 0 & -i\end{bmatrix} \right)^{-1} \left( \hat{A}-ik P^{\perp} \begin{bmatrix} \mathds{1} & 0 \\ 0 & -i\end{bmatrix} \right)$$ 
and 
$$\Cm(\overline{k},-i\overline{k})^* = \left( \hat{A}+ ik  \begin{bmatrix} \mathds{1} & 0 \\ 0 & -i\end{bmatrix}P^{\perp} \right) \left( \hat{A}- ik \begin{bmatrix} \mathds{1} & 0 \\ 0 & -i\end{bmatrix}P^{\perp} \right)^{-1} $$ 
one obtains 
$$ \Cm(k,ik) \begin{bmatrix} \mathds{1} & 0 \\ 0 & i\end{bmatrix} \Cm(\overline{k},-i\overline{k})^*   = \begin{bmatrix} \mathds{1} & 0 \\ 0 & i\end{bmatrix} $$ 
and therefore using Part (i) of Lemma~\ref{cmsym} gives the claim. 

(iii) For the proof of the third part of the lemma decompose 
$$P \pm ik P^{\perp} \begin{bmatrix} \mathds{1} & 0 \\ 0 & -i\end{bmatrix}$$ 
with respect to the orthogonal spaces $\mbox{Ran} P$ and $\mbox{Ran} P^{\perp}$. Denote by 
\begin{eqnarray*}
\left( P^{\perp} \begin{bmatrix} \mathds{1} & 0 \\ 0 & \mp i\end{bmatrix} P^{\perp}\right)^{-1} &\mbox{the inverses of}& P^{\perp} \begin{bmatrix} \mathds{1} & 0 \\ 0 & \mp i\end{bmatrix}  P^{\perp}
\end{eqnarray*}
considered as maps in the Hilbert space $\mbox{Ran} P^{\perp}$. Applying the formula for the Schur complement, see for example \cite{Zhang}, delivers the $k$-independent block-operator matrix representation
\begin{align*}
&\Cm(k,\pm ik)= -\left[ P+ ik P^{\perp}\begin{bmatrix} \mathds{1} & 0 \\ 0 & \mp i\end{bmatrix} \right]^{-1}\left[ P- ik P^{\perp}\begin{bmatrix} \mathds{1} & 0 \\ 0 & \mp i\end{bmatrix}\right]= \\  \\ 
&\begin{bmatrix}
-\left( P^{\perp}  \begin{bmatrix} \mathds{1} & 0 \\ 0 & \mp i\end{bmatrix}  P^{\perp}\right)^{-1} & -\left( P^{\perp}  \begin{bmatrix} \mathds{1} & 0 \\ 0 & \mp i\end{bmatrix}  P^{\perp}\right)^{-1} P^{\perp}  \begin{bmatrix} \mathds{1} & 0 \\ 0 & -i\end{bmatrix}  P \\
 & \\
0 & P
\end{bmatrix}.
\end{align*}
\end{proof}
\subsection{Resonance equation.}
The resonance equation \eqref{EWIndQG} can be rewritten in an analogue way to the one known for Laplacians on finite metric graphs, see for example \cite[Theorem 3.2]{VKRS2006}. For $k,\kappa\neq 0$ such that $\left( A+B I_{n,m}(k,i\kappa)\right)$ is invertible the operator $Z_{n,m}(A,B,k,i\kappa,\underline{a})$ can be rewritten as follows 
\begin{align*}
&Z_{n,m}(A,B,k,i\kappa,\underline{a})\\
= &AX_{n,m}(k,i\kappa) + B Y_{n,m}(k,i\kappa)  \\
=&\left( A+B I_{n,m}(k,i\kappa)  \right)R^+_{n,m}(k,i\kappa) + \left( A-B I_{n,m}(k,i\kappa)  \right)R^-_{n,m}(k,i\kappa) \\
= &\left( A+B I_{n,m}(k,i\kappa)  \right) \left( \mathds{1} - \Cm(k,i\kappa) T_{n,m}(k,i\kappa)   \right), 
\end{align*}
where
\begin{eqnarray*}
R^+_{n,m}(k,i\kappa)=&  \frac{1}{2}\left( X_{n,m}(k,i\kappa) + Y_{n,m}(k,i\kappa) \right), \\
  R^-_{n,m}(k,i\kappa)=&\frac{1}{2}\left( X_{n,m}(k,i\kappa) - Y_{n,m}(k,i\kappa) \right)
\end{eqnarray*}
become more explicit
\begin{align*}
R_{n,m}^+(k,i\kappa,\au)= \begin{bmatrix} R_{n}^+(\au_+,k)  & 0 \\ 0 &  R_{m}^+(\au_-,i\kappa)\end{bmatrix}, && R_{l}^+(\bu,k)=\begin{bmatrix} \mathds{1} & 0 &0 \\ 0 & \mathds{1} &0 \\ 0 &0& e^{ik\bu}   \end{bmatrix},
\end{align*}
\begin{align*}
R_{n,m}^-(k,i\kappa,\au)= \begin{bmatrix} R_{n}^-(\au_+,k)  & 0 \\ 0 &  R_{m}^-(\au_-,i\kappa)\end{bmatrix}, && R_{l}^-(\bu,k)=\begin{bmatrix} 0 & 0 &0 \\ 0 & 0 & \mathds{1}  \\ 0 & e^{ik\bu} & 0   \end{bmatrix} 
\end{align*}
and 
\begin{align*}
T_{n,m}(k,i\kappa,\au)=&R_{n,m}^{-}(k,i\kappa)\left(R_{n,m}^+(k,i\kappa)\right)^{-1} \\
=& \begin{bmatrix} T_{n}(\au_+,k)  & 0 \\ 0 &  T_{m}(\au_-,i\kappa)\end{bmatrix}, && T_{l}(\bu,k)=\begin{bmatrix} 0 & 0 &0 \\ 0 & 0 &e^{ik\bu} \\ 0 &e^{ik\bu}& 0  \end{bmatrix}. 
\end{align*}
Here for $l=n$ one plugs in $\bu=\au_+$ and for $l=m$ one inserts $\bu=\au_-$. Assume that $k\neq 0$ and that the matrices $ A+B I_{n,m}(k,ik)$, $A+B I_{n,m}(k,-ik)$ are invertible. Then one obtains the representations
\begin{align}\label{resonance}
\begin{split}
&Z_{n,m}(A,B, k,\pm ik,\underline{a})= \\
&\left( A+B I_{n,m}(k,\pm ik)  \right) \left( \mathds{1} - \Cm(k,\pm ik) T_{n,m}(k,\pm ik,\au)   \right) R_{n,m}^+(k,\pm ik,\au).
\end{split}
\end{align}

\subsection{The resolvent.}
The Green's function for the equation $\left(-\frac{d^2}{dx^2}-k^2\right)u=f$, $f\in L^2(\R)$, is $\tfrac{i}{2k} e^{ik\abs{x-y}}$. For computing the Green's function for the problem considered here, one has to find appropriate correction terms. These corrections can be expressed in terms of the solutions of the homogeneous problem  $\left(\mp\frac{d^2}{dx^2}-k^2\right)u=0$ on each edge. The Green's function obtained in this way defines an integral operator, which for $k^2\in \C\setminus \R$ is the resolvent operator $(T(A,B)-k^2)^{-1}$. For the computation of the resolvent one follows the guide lines of \cite[Section 4]{VKRS2006}. The notion of integral operators is specified in the following definition borrowed from there.
\begin{definition}[{\cite[Definition 4.1]{VKRS2006}}]
The operator $\mathfrak{K}$ on the Hilbertspace $\mathcal{H}$ is called \textit{integral operator} if for all $j,j^{\prime}\in \Ee \cup \Ie$ there are measurable functions $\mathfrak{K}_{j,j^{\prime}}(\cdot,\cdot)\colon I_j \times I_{j^{\prime}}\rightarrow \C$ with the following properties
\begin{enumerate}
\item $\mathfrak{K}_{j,j^{\prime}}(x_j,\cdot)\varphi_{j^{\prime}}\in L^1(I_{j^{\prime}})$ for almost all $x_j\in I_j$, 
\item $\psi = \mathfrak{K}\varphi$ with 
\begin{equation*}
\psi_j(x_j)= \sum_{j^{\prime}\in \Ee \cup \Ie} \int_{I_{j^{\prime}}} \mathfrak{K}_{j,j^{\prime}}(x_j,y_{j^{\prime}}) \varphi_{j^{\prime}}(y_{j^{\prime}}) d y_{j^{\prime}}.
\end{equation*}
\end{enumerate} 
The $(\Ie \cup \Ee) \times (\Ie \cup \Ee)$ matrix-valued function $(x,y) \mapsto \mathfrak{K}(x,y)$ with 
$$ [\mathfrak{K}(x,y)]_{j,j^{\prime}} = \mathfrak{K}_{j,j^{\prime}}(x_j,y_{j^{\prime}})$$
is called the \textit{integral kernel} of the operator $\mathfrak{K}$.
\end{definition}

\begin{proposition}\label{resolvent}
Let $T(A,B)$ be self-adjoint. Then the resolvent 
\begin{eqnarray*}
R(k^2)=(T(A,B)-k^2)^{-1}, & \mbox{for } k^2\in\C\setminus\R, 
\end{eqnarray*}
is an integral operator with kernel 
\begin{align*}
r(x,y,k)=\begin{cases} 
                r_{n,m}^0(x,y,k,+ik) + r_{n,m}^1(x,y,k,+ik), & k\in \Qe, \\
                r_{n,m}^0(x,y,k,-ik)+ r_{n,m}^1(x,y,k,-ik), & k\in \Pe,
             \end{cases}
\end{align*}
where the free Green's function is given by 
\begin{align*}
r^0_{n,m}(x,y,k,i\kappa)= \begin{bmatrix}  r^0_{n}(x,y,k) & 0 \\ 0 & -r^0_{m}(x,y,i\kappa)\end{bmatrix}W_{n,m}(k,i\kappa),
\end{align*}
where $r^0_{l}(x,y,k)$, $l\in\{n,m\}$ are diagonal matrices with entries
\begin{align*}
\{r^0_{l}(x,y,k)\}_{p,q}=  \delta_{p,q}e^{ik\abs{x_q-y_q}}, 
\end{align*}
\begin{align*}
W_{n,m}(k,i\kappa)= \begin{bmatrix} W_n(k)  & 0 \\ 0 & W_m(i\kappa)  \end{bmatrix}, 
\end{align*}
where $W_n(k)= \frac{i}{2k}\mathds{1}$ with the identity in $\Ke_+$ and $W_m(k)= \frac{i}{2k}\mathds{1}$ with the identity in $\Ke_-$. The correction term is given by
\begin{align*}
r^1(x,y,k,i\kappa)= 
 - \Phi_{n,m}(x,k,i\kappa) G_{n,m}(k,i\kappa,\au) \Phi_{n,m}(y,k,i\kappa)^T W_{n,m}(k,i\kappa),
\end{align*}
where the subscript $T$ denotes the transposed matrix, 
\begin{align*}
G_{n,m}(k,i\kappa,\au)= \mathfrak{X}(k,i\kappa)(R_{n,m}^+(k,i\kappa,\au))^{-1}J_{n,m} 
\end{align*}
and 
\begin{equation*}
\Phi_{n,m}(x,k,i\kappa)= \begin{bmatrix} \Phi_n(x,k) & 0 \\ 0 & \Phi_m(x,i\kappa)  \end{bmatrix} 
\end{equation*}
with
\begin{align*}
\Phi_{l}(x,k)= \begin{bmatrix} \varphi_{\Ee_{l}}(x,k) & 0 & 0 \\ 0 & \varphi_{\Ie_{l}}(x,k) & \varphi_{\Ie_{l}}(x,-k)  \end{bmatrix}. 
\end{align*}
This notation means that for $l=n$ one plugs in $\Ee_{n}=\Ee_+$ and $\Ie_{n}=\Ie_+$ and for $l=m$ one inserts $\Ee_{m}=\Ee_-$ and $\Ie_{m}=\Ie_-$. The blocks $\varphi_{\mathcal{C}}(x,k)$, $\mathcal{C}\in \{\Ee_+,\Ee_-,\Ie_+,\Ie_- \}$ are diagonal matrices with entries $\{\varphi_{\mathcal{C}}(x,k)\}_{l,j\in \mathcal{C}}= \delta_{l,j} \{  e^{ ik x_j} \}$.  
\end{proposition}
The proof is postponed to Appendix~\ref{AppA}. From the explicit formula in Proposition~\ref{resolvent} one reads that the resolvent kernel can be defined for all $k\in \overline{\Qe}$ as the limit of values taken from $\Qe$, except for those $k\in \partial \Qe$ for which $\det Z_{n,m}(A,B,k,ik,\au)=0$. These exceptional vales are called \textit{resonances of $T(A,B)$ in $\overline{\Qe}$}. Analogue for all $k\in \overline{\Pe}$ the resolvent kernel is defined, as the limit of values taken from $\Pe$, except for those $k\in \partial \Pe$ for which $\det Z_{n,m}(A,B,k,-ik,\au)=0$. These are called the \textit{resonances of $T(A,B)$ in $\overline{\Pe}$}. This justifies to denote the equations 
$$\det Z_{n,m}(A,B,k, \pm ik,\au)=0$$
the \textit{resonance equations} for the operator $T(A,B)$. 

\begin{remark}
From the integrability properties of the integral kernel one reads that the resolvent, as a function in $k\in \Qe$ or $k \in \Pe$, admits a meromorphic continuation to $\overline{\Qe}$ or $\overline{\Pe}$. The continuation is possible outside the resonances of $T(A,B)$ in $\overline{\Qe}$ and outside the resonances of $T(A,B)$ in $\overline{\Pe}$. There the resolvent kernel $r(\cdot,\cdot,k)$ defines an operator 
\begin{eqnarray*}
R_{\varepsilon}(k^2)\colon L^2(\Ge, e^{\varepsilon x}dx)\rightarrow L^2(\Ge, e^{-\varepsilon x}dx), & \mbox{for any } \varepsilon >0. 
\end{eqnarray*}
Consequently outside the resonances of $T(A,B)$ the resolvent kernel defines an operator from $L^2_{\mbox{comp}}(\Ge, dx)\rightarrow L^2_{\mbox{loc}}(\Ge, dx)$, where $L^2_{\mbox{comp}}(\Ge, dx)$ denotes the set of compactly supported elements of $\He$ and $L^2_{\mbox{loc}}$ denotes the locally square integrable functions on $(\Ge,\au)$.
\end{remark}
\begin{remark}
The values $k$ and $i\kappa$, are square roots of the spectral parameters $\lambda=k^2$ and $\lambda=-\kappa^2$, respectively. In the explicit formulae involving $k$ and $i\kappa$ it is hidden that two different branches of the complex square root are used simultaneously. To be specific one uses the branch with $\Im \sqrt{\cdot}>0$ and the one with $\Re \sqrt{\cdot}>0$. 
\end{remark}

\section{Scattering}\label{secscat} 
On the one hand scattering theory gives an interpretation of scattering solutions in the context of time--dependent dynamics. On the other hand it gives a description of the absolutely continuous part of an operator in terms of perturbation theory. As general references for mathematical scattering theory the author recommends the books \cite{Yaf} and \cite{Baum}. 

Here the role of the free operator is played by $T(0,\mathds{1})$. This means that in the free system all edges are decoupled and on each positive edge one has the positive Neumann Laplace operator along with the Neumann Laplace operator multiplied by minus one on each negative edge. So, for the self-adjoint operator $T(A,B)$ one considers the scattering pair 
$$\left(T(A,B) ,T(0,\mathds{1})\right).$$ 
The pre--wave operators $W(t)=W(t)\left(T(A,B) ,T(0,\mathds{1})\right)$, where $t\in\R$ are defined by 
\begin{equation*}
W(t)\left(T(A,B) ,T(0,\mathds{1})\right) := e^{itT(A,B)} e^{-it T(0,\mathds{1})}.
\end{equation*}
The strong wave operators $W_{\pm}=W_{\pm}\left(T(A,B) ,T(0,\mathds{1})\right)$ are the strong limits
\begin{equation*}\index{$W(t)$, $W_{\pm}$}
W_{\pm}\left(T(A,B) ,T(0,\mathds{1})\right):= s-\lim_{t\to \pm \infty}W(t)P^{ac}_{T(0,\mathds{1})},
\end{equation*}
where $P^{ac}_{T(0,\mathds{1})}$ is the orthogonal projector onto the absolutely continuous part of $T(0,\mathds{1})$. 
\begin{theorem}\label{CompletnessWO}
Let $T(A,B)$ be self--adjoint. Then the strong wave operators $W_{\pm}$ exist and are complete. The absolutely continuous spectrum $\sigma_{ac}(T(A,B))$ of $T(A,B)$ is 
\begin{eqnarray*}
\sigma_{ac}(T(A,B))= \begin{cases} 
(-\infty,\infty), & \mbox{if } \Ee_+\neq\emptyset \ \mbox{and } \Ee_-\neq\emptyset, \\
[0,\infty), & \mbox{if } \Ee_-=\emptyset \ \mbox{and } \Ee_+\neq\emptyset, \\
(-\infty,0], & \mbox{if } \Ee_+=\emptyset \ \mbox{and } \Ee_-\neq\emptyset,  \\
\emptyset, & \mbox{if } \Ee_+=\emptyset \ \mbox{and } \Ee_-=\emptyset. 
   \end{cases}
\end{eqnarray*}
The multiplicity of $(-\infty,0)$ is $\abs{\Ee_-}$ and the multiplicity of $(0,\infty)$ is $\abs{\Ee_+}$.
\end{theorem}
\begin{proof}
By Proposition~\ref{resolvent} the operator $$(T(A,B)-i)^{-1} - (T(0,\mathds{1})-i)^{-1}$$ is a finite rank operator and from \cite[Theorem 6.5.1]{Yaf} it follows that the strong wave operators exist and are complete. 
\end{proof}

\subsection{Wave operators.}
In the stationary scattering theory the wave operators are calculated in terms of the resolvents of the perturbed operator $T(A,B)$ and of the free operator $T(0,\mathds{1})$. The abelian wave operators $W_{\pm}^{a}=W_{\pm}^{a}\left(T(A,B) ,T(0,\mathds{1})\right)$ can be computed as 
\begin{equation*}
\displaystyle{ W_{\pm}^{a}f =\lim_{\varepsilon \to 0+} - i \varepsilon \int_{-\infty}^{\infty} R(\lambda\mp i\varepsilon) dE_0(\lambda)f  ,}
\end{equation*}\index{$dE_0(\cdot)$}
where $dE_0(\cdot)$ is the absolutely continuous part of the spectral measure of the free operator and $R(\lambda\mp i\varepsilon)= \left( T(A,B) - (\lambda\mp i\varepsilon)\right)^{-1}$, see \cite[Proposition 13.1.1, formula (3)]{Baum}. Since the strong wave operators exist they agree with the abelian wave operators, compare for example \cite{Yaf}. Note that the absolutely continuous part of the free operator $T(0,\mathds{1})$ is related only to the external edges and can be expressed in terms of the $\cos$-transform, see for example \cite[Example 14.8]{W2}. The absolutely continuous part of the spectral measure of the free operator $T(0,\mathds{1})$ is
$$d E_0(\lambda)= \bigoplus_{j\in \Ee} (d E_{0,j}(\lambda)).$$ 
For $j\in \Ee_+$ and $\lambda>0$ one has
\begin{align*}
\left( d E_{0,j}(\lambda)f\right)(x)= 
																	\cos(\sqrt{\lambda} x_j) \frac{d\lambda}{\pi \sqrt{\lambda}} \int_{0}^{\infty} \cos(\sqrt{\lambda} y_j)f(y_j)dy_j. 
\end{align*}
and for $j\in \Ee_-$ and $\lambda<0$  
\begin{align*}
\left( d E_{0,j}(\lambda)f\right)(x)= 
																	\cos(\sqrt{\abs{\lambda}} x_j) \frac{d\lambda }{\pi\sqrt{\abs{\lambda}}} \int_{0}^{\infty} \cos(\sqrt{\abs{\lambda}} y_j)f(y_j)dy_j.                     
\end{align*}
For the cases $j\in \Ee_+$ and $\lambda<0$ and $j\in \Ee_-$ and $\lambda>0$ one has $\left( d E_{0,j}(\lambda)f\right)(x)= 0$. Now substitute $\lambda>0$ by $k^2=\lambda$ with $k>0$ and  $\lambda<0$ by $-\kappa^2=\lambda$ with $\kappa>0$. Then define on each exterior edge $j\in\Ee$ the $\cos$--transform by
$$\displaystyle{  \hat{f}_j(k) = \frac{2}{\pi} \int_{0}^{\infty} \cos(k y_j)f(y_j) dy_j  }$$
and its inverse by
$$\displaystyle{  f_j(k) = \int_{0}^{\infty} \cos(k y_j)\hat{f}(k) dk  }.$$
Set for brevity 
\begin{eqnarray*}
\hat{f}_+(k)= \left\{  \hat{f}_j(k)  \right\}_{j\in \Ee_+}, \ \  \hat{f}_-(\kappa)=  \left\{  \hat{f}_j(\kappa)  \right\}_{j\in \Ee_-} & \mbox{and} & \hat{f}(k,\kappa)= \begin{bmatrix} \hat{f}_+(k) \\ \hat{f}_-(\kappa) \end{bmatrix}. 
\end{eqnarray*}

In the calculation of $W^a_{\pm}$ one can interchange the limit $\varepsilon\to 0+$ and the integration over $dE_0(\lambda)$ according to \cite[Theorem 4.2.4]{Yaf}. After the substitution $k^2=\lambda>0$ one calculates using \cite[Definition 2.7.2]{Yaf} 
\begin{align*}
\lim_{\epsilon \to 0+}-i\epsilon R(k^2\pm i\epsilon) dE_0(k^2)f=   \lim_{\epsilon \to 0+} -i\epsilon\int_{\Ge}r(x,y,\sqrt{k^2\pm i\epsilon}) \begin{bmatrix} \cos(k x)\hat{f}_+(k) \\ 0  \end{bmatrix}   dy \, dk 
\end{align*}
with $$\left\{\cos(k x)\hat{f}_+(k)\right\}_{j\in \Ee_+}= \cos(k x_j)\hat{f}_j(k)_{j\in \Ee_+}$$ and $0$ denotes the zero on the rest of the components. Analogue after substituting $-\kappa^2=\lambda<0$ one computes
\begin{align*}
&\lim_{\epsilon \to 0+}-i\epsilon R(-\kappa^2\pm i\epsilon) dE_0(-\kappa^2)f \\ =&\lim_{\epsilon \to 0+} -i\epsilon\int_{\Ge}r(x,y,\sqrt{-\kappa^2\pm i\epsilon}) \begin{bmatrix} 0 \\ \cos(\kappa x)\hat{f}_-(\kappa)   \end{bmatrix}  dy\, d\kappa
\end{align*}
with $$\left\{\cos(\kappa x)\hat{f}_-(\kappa)\right\}_{j\in \Ee_-}= \cos(\kappa x_j)\hat{f}_j(\kappa)_{j\in \Ee_-}$$ and again $0$ denotes the zero on the rest of the components. This calculation is carried out in Appendix~\ref{AppB}. The result is that the strong wave operators can be computed in terms of the generalized eigenfunctions of $T(A,B)$. Note that $\Ran P^{ac}_{T(0,\mathds{1})}= \He_{\Ee}$. 
\begin{proposition}\label{wo}
For $f\in \He_{\Ee}$ one has
\begin{align*}
W_+f = 
\displaystyle{\sum_{l\in E_+} \int_{0}^{\infty}\varphi_l(\cdot,-k,ik)\hat{f}_l(k) \tfrac{d k}{2}} 
 + \displaystyle{\sum_{l\in E_-}  \int_{0}^{\infty}\varphi_l(\cdot,ik,k)\hat{f}_l(k) \tfrac{d k}{2}}    
\end{align*}
and
\begin{align*}
W_-f =  
\displaystyle{\sum_{l\in E_+} \int_{0}^{\infty}\varphi_l(\cdot,k,ik)\hat{f}_l(k) \tfrac{d k}{2}}  
+\displaystyle{\sum_{l\in E_-}  \int_{0}^{\infty}\varphi_l(\cdot,ik,-k)\hat{f}_l(k) \tfrac{d k}{2}}. 
\end{align*}
\end{proposition}

\begin{remark}\label{U}
From Theorem~\ref{CompletnessWO} it follows that the absolutely continuous part of $T(A,B)$ is the range of $W_{\pm}$ and spectral representations of the absolutely continuous part of $T(A,B)$ are given by the transforms
$$U_{\pm}=W_{\pm}W_{\pm}^*.$$ 
\end{remark}

\subsection{Scattering matrix.}
As the wave operators exist and since they are complete one can define the scattering operator $S=S\left(T(A,B) ,T(0,\mathds{1})\right)$ by  
$$S:=W_+^{\ast}W_-.$$
The scattering operator is a unitary operator on the absolutely continuous subspace of $T(0,\mathds{1})$. From the formulae for the wave operators given in Proposition~\ref{wo} one sees that the wave operators decompose into two parts. One part is related to the positive absolutely continuous spectrum and the other part is related to the negative absolutely continuous spectrum. This observation can be formalized within the concept of local wave operators, see for example \cite[Chapter 2.2.2]{Yaf}. The local wave operators $W_{\pm}\left(T(A,B) ,T(0,\mathds{1}), \Lambda\right)=W_{\pm}\left(\Lambda\right)$ are  
\begin{equation*}
W_{\pm}\left(T(A,B) ,T(0,\mathds{1}), \Lambda\right) := s-\lim_{t\to \pm \infty} e^{itT(A,B)} e^{-it T(0,\mathds{1})} P_0^{ac}(\Lambda),
\end{equation*}
where $P_0^{ac}(\cdot)$ is the spectral projector of the absolutely continuous part of the free operator $T(0,\mathds{1})$ and $\Lambda\subset \sigma_{ac}(T(0,\mathds{1}))$. For $\Lambda_+=[0,\infty)$ and $\Lambda_-=(-\infty,0]$ one has, because $\abs{\Lambda_+ \cap \Lambda_-}=0$, that furthermore
\begin{equation*}
W_{\pm} = W_{\pm}\left(\Lambda_+\right) + W_{\pm}\left(\Lambda_-\right).
\end{equation*}
The local operators $W_{\pm}\left(\Lambda_+\right)$ are
\begin{align*}
W_+\left(\Lambda_+)\right)f &= 
\displaystyle{\sum_{l\in E_+} \int_{0}^{\infty}\varphi_l(\cdot,-k,ik)\hat{f}_l(k) \tfrac{d k}{2}}, \\
W_+\left(\Lambda_-)\right)f &=  \displaystyle{\sum_{l\in E_-}  \int_{0}^{\infty}\varphi_l(\cdot,ik,k)\hat{f}_l(k) \tfrac{d k}{2}}    
\end{align*}
and
\begin{align*}
W_-\left(\Lambda_+\right)f &=  
\displaystyle{\sum_{l\in E_+} \int_{0}^{\infty}\varphi_l(\cdot,k,ik)\hat{f}_l(k) \tfrac{d k}{2}},\\   
W_-\left(\Lambda_-\right)f &= \displaystyle{\sum_{l\in E_-}  \int_{0}^{\infty}\varphi_l(\cdot,ik,-k)\hat{f}_l(k) \tfrac{d k}{2}}. 
\end{align*}
For the scattering operator it follows that
\begin{align*}
S= W_{+}\left(\Lambda_+\right)^{\ast} W_{-}\left(\Lambda_+\right) + W_{+}\left(\Lambda_-\right)^{\ast} W_{-}\left(\Lambda_-\right). 
\end{align*}
Let $U^0$ be a given spectral representation for the free operator, then the corresponding \textit{scattering matrix} is defined as  
$$S(\cdot)= U_0 W_+^*W_- U_0^*.$$ 
The scattering matrix is unitary. Here one chooses $U_0$ to be the diagonal block-operator matrix
\begin{eqnarray*}
U_0= \begin{bmatrix} U_0^+ & 0 \\ 0 & U_0^-  \end{bmatrix}, &\mbox{where} & U_0^+f=\hat{f}_+, \ \mbox{and}  \ U_0^-f=\hat{f}_-.
\end{eqnarray*}
This is the $\cos$--transform for the absolutely continuous part of the free operator $T(0,\mathds{1})$, that already appeared in the context of stationary scattering theory. As a consequence the scattering matrix  is given as a diagonal block operator matrix. Using Proposition~\ref{wo} the scattering matrix can be computed in terms of the coefficients of the generalized eigenfunctions. In the following this is carried out in detail.

Write $\varphi_l(\pm k,i\kappa, x)$, $l\in\Ee$, into the $\abs{\Ee} \times \abs{\Ee\cup \Ie}$--matrix $\Phi(x,\pm k,ik)$ with entries 
\begin{align*}
\left\{\Phi(x,\pm k,ik)\right\}_{j,l} := \left\{\varphi_l(x,\pm k,i\kappa)\right\}_j.
\end{align*}
The restrictions of $\Phi(x,\pm k,ik)$ to $\Ee_+$ are
\begin{align*}
\left\{\Phi(x,+k,ik)\right\}_{j,l\in\Ee_+} = \delta_{lj} e^{-ikx_j} +  e^{ikx_l}\chi_{j,l}(k,ik) 
\end{align*}
and
\begin{align*}
\left\{\Phi(x,-k,ik)\right\}_{j,l\in\Ee_+} = \delta_{lj} e^{ikx_j} +   e^{-ikx_l} \chi_{j,l}(-k,ik).  
\end{align*}
Consider the restriction of $\chi(k,ik)$ to the external edges and denote it by $\chi_{\Ee,\Ee}(k,i\kappa)$, which is a $\abs{\Ee}\times \abs{\Ee}$--matrix with entries $$\{\chi(k,i\kappa)_{\Ee,\Ee}\}_{j,l\in \Ee}=\left\{\chi(k,i\kappa)\right\}_{j,l\in \Ee}.$$ 
With respect to the division into positive and negative external edges one obtains the block structure 
\begin{eqnarray}\label{Sm}\index{$\chi_{\pm\pm}(k,i\kappa)$, $\chi_{\Ee,\Ee}(k,i\kappa)$}
\chi_{\Ee,\Ee}(k,i\kappa)=\begin{bmatrix}
\chi_{++}(k,i\kappa) & \chi_{+-}(k,i\kappa) \\
\chi_{-+}(k,i\kappa) & \chi_{--}(k,i\kappa)
\end{bmatrix},
\end{eqnarray}
where $\chi_{p,q}(k,i\kappa)$ with $p,q\in \{+,-\}$ are the matrices with entries  
$$\{\chi(k,i\kappa)_{p,q}\}_{j\in \Ee_{p}, l\in \Ee_{q}}=\left\{\chi(k,i\kappa)\right\}_{j\in \Ee_{p}, l\in \Ee_{q}}.$$ 
Since the functions $e^{-ikx}$ and $e^{ikx}$ are on each edge linearly independent there exists a $\abs{\Ee_+}\times \abs{\Ee_+}$--matrix $C_{+}(k)$ with entries $\{C_{+}(k)\}_{l,j\in \Ee_+}=c_{lj}(k)$ such that 
$$\Phi(\cdot,k,ik)C_+(k)=\Phi(\cdot,-k,ik)$$ 
holds on the positive external edges. The functions $c_{lj}(k)\varphi_l(\cdot,k,ik)$ are generalized eigenfunctions as well and therefore their linear combinations
$$\varphi_j(\cdot,-k,ik)= \sum_{l\in \Ee_+}c(k)_{lj}\varphi_l(\cdot,k,ik), \ \ \mbox{for} \ x\in \Ee_+$$ 
are generalized eigenfunctions, too. Recall that the functions $\varphi_l(\cdot,k,ik)$ and \\ $\varphi_l(\cdot,-k,ik)$ are defined uniquely up to a set of measure zero. Hence the relation 
$$\varphi_j(\cdot,-k,ik)= \sum_{l\in \Ee_+}c(k)_{lj}\varphi_l(\cdot,k,ik)$$ 
carries over from the external parts to the whole graph $(\Ge,\au)$. Comparing the coefficients on the external edges yields 
\begin{eqnarray*}
C_{+}(k) \chi_{++}(k,ik)=\mathds{1}_{\Ee_+} & \mbox{and hence} & C_{+}(k)^{-1}= \chi_{++}(k,ik).
\end{eqnarray*}
Analogously one obtains 
$$\varphi_l(\cdot,ik,k)= \sum_{l\in \Ee_-}c_{lj}(k)\varphi_l(\cdot,ik,-k),$$
for an appropriate $\abs{\Ee_-}\times \abs{\Ee_-}$--matrix $C_{-}(k)$ with entries  $\{C_-(k)\}_{j,l\in \Ee_-}=c_{lj}(k)$. Considering the restriction of $\varphi_l(\cdot,ik,\pm k)$ to $\Ee_-$ yields
\begin{eqnarray*}
\Phi(\cdot,ik,-k)C_-(k)=\Phi(\cdot,ik,k) 
\end{eqnarray*}
and hence
\begin{eqnarray*}
C_-(k) \chi_{--}(ik,k)=\mathds{1}_{\Ee_-}  & \mbox{and} & C_-(k)^{-1}= \chi_{--}(ik,k).
\end{eqnarray*}
This can be used to express $W_-$ in terms of $W_+$ as follows,
\begin{eqnarray}\label{wm}
W_+  M_C U_0 f = W_- f, 
\end{eqnarray}
where $M_C$ denoted the multiplication operator with $C(\cdot)$, 
\begin{eqnarray*}
C(k)=\begin{bmatrix} C_+(k) & 0 \\ 0 & C_-(k) \end{bmatrix} &\mbox{and hence }  M_C\hat{f}_{\pm}(k)= C_{\pm}(k)\hat{f}_{\pm}(k).
\end{eqnarray*}
As the absolutely continuous parts of $T(\mathds{1},0)$ and $T(A,B)$ are unitarily equivalent one can introduce the wave matrices using for example the transform $U_-$ from Remark~\ref{U}. Here the wave matrices $w_+$ and $w_-$  are defined by 
\begin{eqnarray*}
w_{+}(k)\hat{f}(k)= U_- W_{+}U_0^*\hat{f}(k)& \mbox{and} & w_{-}(k)\hat{f}(k)= U_- W_{-}U_0^*\hat{f}(k),
\end{eqnarray*}
and they are unitary maps from $\Ran U_0$ to Ran $U_-$. With the above equation \eqref{wm} one obtains using $U_0^{\ast}U_0 f = U_0^{\ast}\hat{f}$ that
\begin{eqnarray*}
w_+^*w_-\hat{f}_+= M^*_{C_+}\hat{f}_+ &\mbox{and} & w_+^*w_-\hat{f}_-= M^*_{C_-}\hat{f}_-,
\end{eqnarray*} 
where $M^*_{C_\pm}$ is the multiplication operator with $C^{\ast}_{\pm}(\cdot)$. By the definitions of $w_+$ and $w_-$ one has for the scattering matrix  
\begin{eqnarray*}
S(k^2)\hat{f}(k)=w_+^*(k)w_-(k)\hat{f}(k)
\end{eqnarray*} 
and hence $C_+(k)$ is unitary for almost all $k>0$ and as $C_+(k)^{-1}=\chi_{++}(k,ik)$ one concludes that $C_+(k)^*=\chi_{++}(k,ik)$. Analogue one obtains $C_-(k)^*=\chi_{--}(ik,k)$. This is summarized in 

\begin{theorem}\label{smatrix}
Let $T(A,B)$ be self--adjoint. Then the scattering matrix for the pair \\ $\left(T(A,B),T(0,\mathds{1})\right)$ is a $\abs{\Ee_+} \times \abs{\Ee_+}$--matrix for positive energies $\lambda>0$  and a $\abs{\Ee_-} \times \abs{\Ee_-}$--matrix for negative energies $\lambda<0$,   
\begin{align*}\index{$S(\lambda)$}
S(\lambda)= \begin{cases} \chi_{++}(\sqrt{\lambda}), & \lambda>0,  \\
                    \chi_{--}(i\sqrt{\abs{\lambda}}), & \lambda<0,
      \end{cases}
\end{align*}
where $\chi_{++}$ and $\chi_{--}$ have been defined in equation \eqref{Sm} as restrictions of the coefficient $\chi$. 
\end{theorem}
Recall that the coefficient $\chi$ has been defined in \eqref{chi}. Especially one obtains, applying Theorem~\ref{smatrix} to star-graphs, the non--obvious fact that the blocks $\Cm_{++}(k,ik)$ and $\Cm_{--}(ik,k)$ of the coefficient matrix $\Cm(k,ik)$ defined in equation \eqref{Cm}, are unitary for all $k>0$, except a finite set. The relevance of Theorem~\ref{smatrix} arises from the fact that it justifies to read the scattering properties of the system directly from the coefficients of the generalized eigenfunctions. This is common when considering self--adjoint Laplacians on graphs, but it is not self--evident, as the example of a Schr\"{o}dinger operator on the real line with a step potential shows. In this case one obtains a relation between the coefficients of the generalized eigenfunctions and the scattering matrix as well, but both do not agree in general, compare for example the article \cite{Guillope}, which is also helpful for the understanding of the scattering problem discussed here. 

\begin{example}\label{ex2IndefQG}
Consider the graph that consists of two external edges $\Ee_+=\{1,2\}$ which are connected by an internal edge $I_-=\{3\}$ of length $\au=\{a\}$. This means that one has the two vertices $\partial (1)=\partial_-(3)$ and $\partial (2)=\partial_+(3)$. The boundary conditions imposed on each vertex are the standard boundary conditions from Example~\ref{ex1IndQG}. They are encoded in the matrices
\begin{eqnarray*} A=
\begin{bmatrix} 
1 & -1 & 0 &0 \\
0& 0 & 0 & 0\\
 0 &0 &1 & -1 \\
0& 0 & 0 & 0
\end{bmatrix} & \mbox{and}& B= 
\begin{bmatrix} 
0& 0 & 0 & 0\\
1 & -1 & 0 &0 \\
0& 0 & 0 & 0 \\
0 &0 &1 & -1 
\end{bmatrix}.
\end{eqnarray*}
Hence one obtains 
\begin{align*}
X_{2,2}(k,ik,\au)=
\begin{bmatrix} 
1 & 0 & 0 &0 \\
0& 1 & 1 & 0\\
 0 &e^{-k a} & e^{k a} & 0 \\
0& 0 & 0 & 1
\end{bmatrix}, 
Y_{2,2}(k,ik,\au)= 
\begin{bmatrix} 
ik & 0 & 0 &0 \\
0& -k & +k & 0\\
 0 &k e^{-k a} & -k e^{k a} & 0 \\
0& 0 & 0 & ik
\end{bmatrix}
\end{align*}
and with $Z_{2,2}(A,B,k,ik,\au)= AX_{2,2}(k,ik,\au) +BY_{2,2}(k,ik,\au)$ one obtains the scattering matrix for $k>0$
\begin{align*} S(k^2)= \begin{bmatrix} 
s_{11}(k) & s_{12}(k)\\
s_{21}(k)& s_{22}(k) 
\end{bmatrix} 
= \begin{bmatrix} 
 i \tanh(ak)   & \frac{1}{\cosh(ak)}\\
\frac{1}{\cosh(ak)} &   i \tanh(ak)  
\end{bmatrix}.
\end{align*}
The absolutely continuous spectrum is $[0,\infty)$ with multiplicity two. The negative eigenvalues are the zeros of  
$$\det Z_{2,2}(A,B,ik,k,\au)= i k^2 \cos(ak)$$ 
and there are no embedded eigenvalues. Hence the pure point spectrum 
$$\sigma_{pp}(T(A,B))=\left\{ -\left(\frac{((2m-1)\pi)^2}{2 a}\right)^2 \bigg\vert \, m\in \mathds{N} \right\}$$ 
consists of infinitely many negative eigenvalues with multiplicity one which accumulate at $-\infty$. The kernel is zero and there are no further resonances. 
\end{example} 
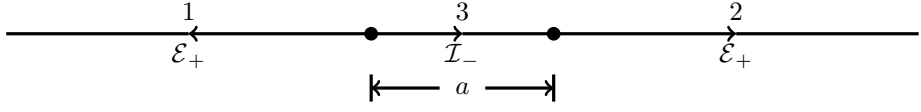
\begin{figure}
  \centering
\begin{tikzpicture}[scale=0.6]
\fill[black] (0,0) circle (1ex);
\fill[black] (-4,0) circle (1ex);
\draw[->, black, very thick] (0,0) -- (4,0);
\draw[black, very thick] (4,0) -- (8,0);
\draw (4,0.5) node 
{2};
\draw (4,-0.5) node 
{$\Ee_+$};
\draw[<-, black, very thick] (-8,0) -- (-4,0);
\draw[black, very thick] (-12,0) -- (-8,0);
\draw (-8,0.5) node 
{1};
\draw (-8,-0.5) node 
{$\Ee_+$};
\draw[->, black, very thick] (-4,0) -- (-2,0);
\draw[black, very thick] (-2,0) -- (0,0);
\draw (-2,0.5) node 
{3};
\draw (-2,-0.5) node 
{$\Ie_-$};

\draw[<-, black, very thick] (-4,-1.2) -- (-2.5,-1.2);
\draw[black, very thick] (-4,-0.9) -- (-4,-1.5);
\draw[black, very thick] (0,-0.9) -- (0,-1.5);
\draw[->,black, very thick] (-1.5,-1.2) -- (0,-1.2);
\draw (-2,-1.2) node 
{$a$};

\end{tikzpicture}
  \caption[The graph described in Example~\ref{ex2IndefQG}, Chapter~\ref{IndefiniteQG}.]{The graph described in Example~\ref{ex2IndefQG}.}
\end{figure}
\begin{example}\label{S1}
For the situation considered in Example~\ref{exCm} one reads from Theorem~\ref{smatrix} that the scattering matrix is 
\begin{align*}
S(\lambda)=\begin{cases} +i , & \lambda>0, \\ -i, & \lambda<0. \end{cases} 
\end{align*}
The absolutely continuous spectrum is the whole real line and there are no eigenvalues. 
\end{example}

\begin{example}
Consider a star graph with three edges, $\Ee_+=\{1,2\}$ and $\Ee_-=\{3\}$ matched together by the boundary conditions
\begin{eqnarray*}
A=\begin{bmatrix}
1 & -1 & 0 \\
0 & 1 & -1 \\
0 & 0 & 0 
\end{bmatrix} & \mbox{and} & B=\begin{bmatrix}
0 & 0 & 0 \\
0 & 0 & 0 \\
1 & 1 & -1 
\end{bmatrix}.
\end{eqnarray*}
Note that these local boundary conditions are the ones from Example~\ref{ex1IndQG}. The operator $T(A,B)$ has absolutely continuous spectrum $(-\infty,\infty)$, where $(0,\infty)$ has multiplicity two and $(-\infty,0)$ has multiplicity one. Since 
\begin{equation*}
\det Z_{2,1}(A,B,k,i\kappa)= \kappa + 2 ik,
\end{equation*}
there are no eigenvalues and no resonances. The local coefficient matrix \\ $\Cm(k,ik)=\Cm$, 
\begin{eqnarray*}
\Cm=\begin{bmatrix} -1/5+(2/5)i & 4/5+(2/5)i & 2/5-(4/5)i \\  4/5+(2/5)i &  -1/5+(2/5)i & 2/5-(4/5)i \\ 4/5+(2/5)i & 4/5+(2/5)i  & -3/5-(4/5)i \end{bmatrix}
\end{eqnarray*}
is $k$--independent and the scattering matrix is given by the corresponding blocks of $\Cm= \chi_{\Ee,\Ee}$,
\begin{align*}
S(\lambda)=\begin{cases} \begin{bmatrix} -1/5+(2/5)i & 4/5+(2/5)i \\  4/5+(2/5)i &  -1/5+(2/5)i  \end{bmatrix} , & \lambda>0, \\ -3/5-(4/5)i, & \lambda<0. \end{cases} 
\end{align*}
\end{example}

Looking at the generalized eigenfunctions and taking into account Theorem~\ref{smatrix} one can try an interpretation of the coefficient matrix $\Cm(k,ik)$. Let $\Ge$ be a star graph and $k>0$, then 
\begin{itemize}
\item $\Cm_{++}(k,ik)$ is the scattering matrix, that is the entries describe the transmitted and reflected parts of an incoming wave (scattering matrix for positive energies), 
\item $\Cm_{+-}(k,ik)$ can be interpreted as depth of indentation of an incoming wave with positive energy $k^2$ into the negative part of the graph. 
\end{itemize}
For negative energies $-k^2<0$ the matrices $\Cm_{--}(ik,-k)$ and $\Cm_{-+}(ik,-k)$ admit analogue interpretations.

\subsection{Time--dependent Problems.}
A difficult topic is to consider time--dependent problems involving indefinite operators. Thinking of applications to solid-state physics, one considers the Schr\"{o}dinger equation. In the effective mass approximation, the effective mass tensor can be negative, too. For any self--adjoint operator one can now give the solutions of the initial value problem
\[ \left\{ \begin{array}{ll}
         \left( i \frac{\partial}{\partial t } - T(A,B) \right) u(x,t)=0, \\ 
        u(\cdot,0)=u_0 \end{array} \right. 
        \] 
in terms of a unitary group 
\begin{eqnarray*}
u=U(t)u_0, &\mbox{where} & U(t)=e^{-it T(A,B)}
\end{eqnarray*}
is acting in $\He$. The fact that the spectrum is not semi--bounded is no obstacle for the construction of this group, similar to the situation, when considering the Dirac equation.

Considering the wave equation or the diffusion equation changes this feature completely. The lower bound on the spectrum of a semi--bounded operator can be interpreted as a measure for the stability of the system. Problems that arise when the spectrum is neither bounded from below nor from above can be avoided by projecting away the critical parts of the spectrum. Assume that $\Ee_-=\emptyset$ and consider a self-adjoint operator $T(A,B)$ on such a graph. Then $T(A,B)$ has only positive absolutely continuous spectrum. Denote by $T_{ac}(A,B)$ the restriction of $T(A,B)$ onto its absolutely continuous subspace. For the wave equation according to \cite[Chapter 10.3]{Baum} with 
$\Be=T_{ac}(A,B)^{1/2}$ the solution of 
\begin{align*}
    \left(  \frac{\partial^2}{\partial t^2 } + T_{ac}(A,B) \right) u(x,t)=0, && u(0)=u_0, && \frac{\partial}{\partial t }u = v_0
\end{align*}
is given by the group 
\begin{align*}
\begin{bmatrix} u(t) \\ u^{\prime}(t)  \end{bmatrix} = \begin{bmatrix} \cos(\Be t) & \Be^{-1} \sin(t\Be) \\ -\Be \sin (t\Be) & \cos(t\Be)   \end{bmatrix} \begin{bmatrix} u_0 \\ v_0  \end{bmatrix},
\end{align*}
acting in an appropriate Hilbert space. All entries are known in terms of the spectral theorem using the spectral transform of $T_{ac}(A,B)$, which in turn is explicitly given in terms of generalized eigenfunctions, compare Remark~\ref{U} and Proposition~\ref{wo}.

Alternatively one can consider a piecewise defined wave equation
\begin{align*}
\begin{cases}  +  \left(  \frac{\partial^2}{\partial t^2 } - \frac{d^2}{dx^2} \right) u(x,t)=0, & x\in \Ge_+, \\
       -\left(  \frac{\partial^2}{\partial t^2 } - \frac{d^2}{dx^2} \right)u(x,t)=0 , & x\in \Ge_-
\end{cases}
\end{align*}
with appropriate initial conditions. This would correspond to a completely reversed dynamics on the negative part compared to the positive one. Such problems -- although very interesting -- are beyond the scope of the approach presented here. 

\subsection{Gluing formula for the scattering matrix.}
   
Let be given two connected metric graphs. Then one can construct a single connected metric graph by gluing the two original graphs together along external edges. Each two external edges that are glued together in pairs, become an internal edge of a certain length. To be more precise let be given two metric graphs $(\Ge_1,\au_1)$, where $\Ge_1= (V_1,\Ie_1,\Ee_1, \partial_1)$ and $(\Ge_2,\au_2)$, where $\Ge_2= (V_2,\Ie_2,\Ee_2, \partial_2)$. Furthermore let be given two subsets $\tilde{\Ee_1}\subset \Ee_1$ and $\tilde{\Ee_2}\subset \Ee_2$ of their external edges with $\abs{\tilde{\Ee_1}}=\abs{\tilde{\Ee_2}}$, and a bijective identification map 
$$G \colon \tilde{\Ee_1}  \rightarrow \tilde{\Ee_2}.$$  
Then define a new graph 
$$ \Ge := \Ge_1 \circ_G \Ge_2 $$\index{$\Ge_1 \circ_G \Ge_2$}
is the graph $\Ge= (V,\Ie,\Ee, \partial)$ with $V=V_1\cup V_2$ and $\Ee= \left(\Ee_1\setminus \tilde{\Ee}_1\right) \cup \left(\Ee_2\setminus \tilde{\Ee}_2\right)$. The internal edges are $\Ie=\Ie_1 \cup \Ie_2 \cup \Ie_G$, where $\Ie_G$  are the new edges connecting the two graphs. One has $\abs{\Ie_G}=\abs{\tilde{\Ee}_1}$. Any element $e\in \tilde{\Ee}_1$ becomes an element of $i=i_e\in\Ie_G$ and one sets $\partial(i_e)=(\partial(e),\partial(G(e)))$. Assigning to each new internal length edge the length $a_{i}$, which is written into $\au_G=\{a_{i}\}_{i\in \Ie_G}$ and keeping the edge lengths of $\Ie_1$ and $\Ie_2$ one obtains a new metric graph $(\Ge,\au)$.

Let there be furthermore two self--adjoint Laplace operators defined on each of the metric graphs $(\Ge_j,\au_j)$, $j=1,2$. Then their boundary conditions define naturally a self-adjoint operator on the new graph $\Ge=\Ge_1 \circ_G \Ge_2$, constructed by gluing together the two original ones. One can consider first the scattering matrices of the operators defined on each of the components separately. Then the scattering matrix of the operator defined on the new graph can be computed in terms of the original data of the scattering matrices and the data used in the gluing construction. This iterative process is known as the \textit{star product}. The history of gluing formulas for scattering matrices of Laplacians on networks goes back to Redheffer, see \cite{Redheffer1961, Redheffer1962}. The star product has been generalized by Kostrykin and Schrader to self--adjoint Laplacians on metric graphs, see \cite{VKRS1999, KS2001}. The considerations presented here are based on these last mentioned works. For the history of the star products and factorizations of the scattering matrix see also the references quoted therein. Such gluing formulae are only known for one dimensional (singular) spaces. For higher dimensions such formulae are not known in general.

Here a gluing formula is presented for the case of two graphs which have equal numbers of negative external edges and all of them are glued together in pairs. This allows understanding of the scattering properties of systems, where the negative part of the leading coefficient is only compactly supported in terms of the generalized star product. 

Consider two finite metric graphs $(\Ge^1,\au^1)$ and $(\Ge^2,\au^2)$ with $\abs{\Ee_-^1}= \abs{\Ee_-^2}$, and on each of these graphs the self--adjoint operators $T(A_1,B_1)$ and $T(A_2,B_2)$, respectively. In addition let there be a bijective identification $G\colon \Ee_-^1 \rightarrow \Ee_-^2$. One considers now the graph $\Ge=\Ge_1 \circ_G \Ge_2$ with $\Ie_G\subset \Ie_-$ and lengths $a_i>0$ for $i\in\Ie_G$. The boundary conditions $(A_1,B_1)$ and $(A_2,B_2)$ act only on $V_1$ or on $V_2$, respectively. Therefore one can impose on $\Ge$ the boundary conditions defined by 
\begin{eqnarray*}
A=\begin{bmatrix} A_1 & 0 \\ 0 & A_2 \end{bmatrix} &\mbox{and} & B=\begin{bmatrix} B_1 & 0 \\ 0 & B_2 \end{bmatrix},
\end{eqnarray*}
which are given with respect to the division $V=V_1 \dot\cup V_2$. Hence $T(A,B)$ is self-adjoint on $\Ge$ for any choice of $\au_G$.

\begin{figure}
  \centering
   \subfigure[Two disconnetcted star graphs and the identification map $G$]{
\begin{tikzpicture}[scale=0.4]
\fill[black] (0,3) circle (1ex);
\draw (-1,4) node 
{$\Ge^1$};

\draw[->, black, dashed, very thick] (0,3) -- (4,3);
\draw[black, dashed,  very thick] (4,3) -- (8,3);
\draw (4,2.3) node 
{$\Ee_-^1$};
\draw[<-, black, very thick] (-4,3) -- (0,3);
\draw[black, very thick] (-8,3) -- (-4,3);
\draw (-4,2.3) node 
{$\Ee_+^1$};

\draw[->, black, very thick] (3.1,2) -- (3.1,0.6);
\draw (2.2,1.5) node 
{$G$};

\fill[black] (8,0) circle (1ex);
\draw (9,1) node 
{$\Ge^2$};

\draw[->, black, very thick] (8,0) -- (12,0);
\draw[black, very thick] (12,0) -- (16,0);
\draw (12,-0.7) node 
{$\Ee_+^2$};
\draw[<-, black, dashed, very thick] (4,0) -- (8,0);
\draw[black, dashed, very thick] (0,0) -- (4,0);
\draw (4,-0.7) node 
{$\Ee_-^2$};
\end{tikzpicture}

}
 \subfigure[Connected graph after gluing construction with internal edge length $a_G$]{
 \begin{tikzpicture}[scale=0.4]
\fill[black] (8,-3) circle (1ex);
\fill[black] (0,-3) circle (1ex);
\draw (3,-2) node 
{$\Ge=\Ge_G$};

\draw[->, black, very thick] (8,-3) -- (12,-3);
\draw[black, very thick] (12,-3) -- (16,-3);
\draw (12,-3.7) node 
{$\Ee_+$};
\draw[->, black, dashed, very thick] (0,-3) -- (4,-3);
\draw[black, dashed, very thick] (4,-3) -- (8,-3);
\draw (4,-3.7) node 
{$\Ie_-$};

\draw[<-, black, very thick] (0,-4.7) -- (3,-4.7);
\draw[black, very thick] (0,-4.4) -- (0,-5);
\draw[black, very thick] (8,-4.4) -- (8,-5);
\draw[->,black, very thick] (5,-4.7) -- (8,-4.7);
\draw (4,-4.7) node 
{$a_G$};

\draw[<-, black, very thick] (-4,-3) -- (0,-3);
\draw[black, very thick] (-8,-3) -- (-4,-3);
\draw (-4,-3.7) node 
{$\Ee_+$};

\end{tikzpicture}
}
  \caption{Two graphs glued along the negative external edges.}
\end{figure}
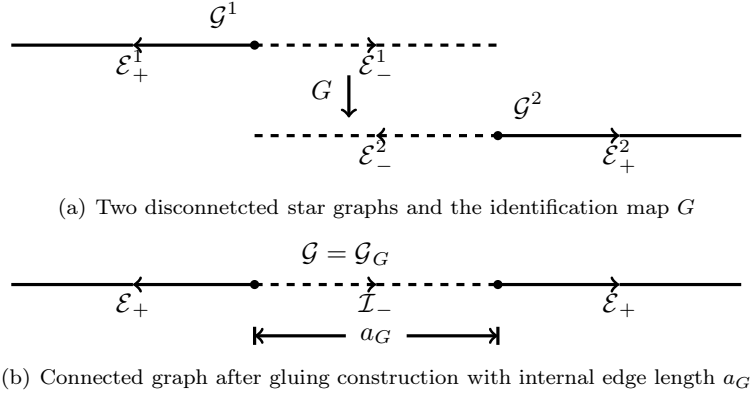

Let $\chi^j(k,ik)$ be the coefficients of the generalized eigenfunctions of $T(A_j,B_j)$ on $(\Ge^j,\au^j)$, for $j=1,2$. With respect to the decomposition into positive and negative external edges one writes as in equation \eqref{Sm}
\begin{align*}
\chi_{\Ee^1,\Ee^1}^1(k,ik)=& \begin{bmatrix}\chi_{++}^1(k,ik) & \chi_{+-}^1(k,ik) \\ \chi_{-+}^1(k,ik) & \chi_{--}^1(k,ik)\end{bmatrix}
\end{align*}
and 
\begin{align*}
\chi^2_{\Ee^2,\Ee^2}(k,ik)=& \begin{bmatrix}\chi_{++}^2(k,ik) & \chi_{+-}^2(k,ik) \\ \chi_{-+}^2(k,ik) & \chi_{--}^2(k,ik)\end{bmatrix}.
\end{align*}
The operator $T(A,B)$ on the metric graph $(\Ge,\au)$, obtained by gluing together both along the negative edges, with new lengths $\au_G$, has only positive absolutely continuous spectrum and its scattering matrix is explicitly computable in terms of $\chi_{\Ee^1,\Ee^1}^1(k,ik)$ and $\chi_{\Ee^2,\Ee^2}^2(k,ik)$, assuming some compatibility properties. To make these compatibility assumptions more precise one
denotes the critical sets where the generalized star product is a priori not defined by
\begin{align*}
\Xi_1(\chi^1,\chi^2, \au):=&\left\{ k>0 \mid \det [\mathds{1}- e^{-k\au} \chi_{--}^1(k,ik)e^{-k\au}\chi_{--}^2(k,ik)]=0 \right\},\\
\Xi_2(\chi^1,\chi^2, \au):=&\left\{ k>0 \mid \det [\mathds{1} - e^{-k\au}\chi_{--}^2(k,ik)e^{-k\au}\chi_{--}^1(k,ik)]=0 \right\}
\end{align*}
and by 
\begin{align*}
\Theta(\chi^1,\chi^2) = \left\{k>0 \mid k \ \mbox{singularity of $\chi^1(k,ik)$ or $\chi^2(k,ik)$}\right\}.  
\end{align*}
The generalized star product is defined in the article \cite{KS2001} and is denoted by $\ast_p$. The gluing formula is described in the following

\begin{proposition}\label{glue}
Let $T(A_1,B_1)$ and $T(A_2,B_2)$ be self--adjoint operators on the graph $(\Ge^1,\au^1)$ and on $(\Ge^2,\au^2)$, respectively, where $p=\abs{\Ee_-^1}=\abs{\Ee_-^2}$. Let furthermore be $T(A,B)$ and $(\Ge,\au)$ as described above. Then for all $k>0$ with $ k\notin\{ \Xi_1(\chi^1,\chi^2, \au) \cup \Xi_2(\chi^1,\chi^2, \au)\cup \Theta(\chi^1,\chi^2)\}$ the scattering matrix of the pair $\left(T(A,B), T(0,\mathds{1})\right)$ is given in terms of the generalized star product as
\begin{align*}
S(k^2)= \chi^1_{\Ee,\Ee}(k,ik) \ast_p V(\au_G,k)  \chi_{\Ee,\Ee}^2(k,ik), && V(\au_G,k)=\begin{bmatrix} e^{-k \au_G} & 0 \\ 0 & \mathds{1} \end{bmatrix}.
\end{align*}
Carrying this out gives 
\begin{align*}
S(k^2)=& \begin{bmatrix}s_{11}(k) & s_{12}(k) \\ s_{21}(k) & s_{22}(k)\end{bmatrix},
\end{align*}
where the blocks $s_{ij}(k)$, $i,j\in\{1,2\}$ of the scattering matrix are matrices with entries 
$$\left\{s_{ij}(k)\right\}_{n,m}=\left\{(s_{n,m}(k))\right\}_{n\in \Ee_+^i, m\in \Ee_+^j}.$$ 
One has
\begin{align*}
s_{11}=&\chi^1_{++} + \chi^1_{+-} e^{-k\au_G} \chi^2_{--} e^{-k\au_G}[\mathds{1} - \chi^1_{--}e^{-k\au_G}\chi^2_{--}e^{-k\au_G}]^{-1}\chi^1_{-+}, \\
s_{21}=& \chi^2_{+-}e^{-k\au_G}[\mathds{1}- \chi^1_{--}e^{-k\au_G}\chi^2_{--}e^{-k\au_G}]^{-1}\chi^1_{-+},\\
s_{12}=& \chi^1_{+-}[\mathds{1}- e^{-k\au_G}\chi^2_{--}e^{-k\au_G}\chi^1_{--}e^{-k\au_G}]^{-1}e^{-k\au_G}\chi^2_{-+},\\
s_{22}=& \chi^2_{++} + \chi^2_{+-} e^{-k\au_G} \chi^1_{--}e^{-k\au_G}[\mathds{1} - e^{-k\au_G}\chi^2_{--}e^{-k\au_G}\chi^1_{--}]^{-1}e^{-k\au_G} \chi^2_{-+},
\end{align*}
where the $k$--dependence is omitted and $\chi^j_{\pm,\pm}$ is short for $\chi^j_{\pm,\pm}(k,ik)$, $j=1,2$.
\end{proposition}
The proof is completely analogue to \cite{KS2001} and omitted here. It is based on the fact that the generalized eigenfunctions of the new problem can be obtained as linear combinations of the generalized eigenfunctions of the two original scattering problems. The restriction imposed on the energies for which the gluing formula is valid is due to the fact that, in contrast to the case of self--adjoint Laplacians on graphs, the relation between eigenvalues and resonances has not been clarified yet.

The formula in Proposition~\ref{glue} shows that for positive energies the matrix $V(\au,k)$ damps the scattering waves on $\Ie_G$ down. Let be $k_n>0$ with $$k_n\notin\{ \Xi_1(\chi^1,\chi^2, \au) \cup \Xi_2(\chi^1,\chi^2, \au)\cup \Theta(\chi^1,\chi^2)\}, \quad n\in \N,$$
and $k_n \to \infty$ for $n\to\infty$ and assume that the limits 
\begin{eqnarray}\label{SDiagonalblocks}
\lim_{n\to \infty} \chi^1_{++}(k_n) & \mbox{and}  &  \lim_{n\to \infty} \chi^2_{++}(k_n)
\end{eqnarray}
exist. As a direct consequence of Proposition~\ref{glue} one obtains that then
\begin{eqnarray*}
\lim_{n\to \infty}s_{12}(k_n)=0, &  &  \lim_{n\to \infty}s_{21}(k_n) = 0
\end{eqnarray*}
and 
\begin{eqnarray*}
\lim_{n\to \infty}s_{11}(k_n)=\lim_{n\to \infty} \chi^1_{++}(k_n), &  &  \lim_{n\to \infty}s_{22}(k_n) = \lim_{n\to \infty} \chi^2_{++}(k_n)
\end{eqnarray*}
which shows that the scattering matrix resembles for large energies the direct sum of the scattering matrices of the two building stones. However it is not clear that the limits in \eqref{SDiagonalblocks} do always exist.

\begin{remark}
Gluing together positive as well as negative edges can be done in two steps. First one glues the negative edges together and then one considers this graph and an auxiliary graph consisting of two edges with continuity boundary conditions. This is used to connect two positive edges. More precisely after having constructed $$S(k^2)= \chi_{\Ee,\Ee}^1(k,ik)\ast_p V(\au_G,k)  \chi^2_{\Ee,\Ee}(k,ik)$$ according to Proposition~\ref{glue} one considers the auxiliary graph $\Ge_{aux}=(V,\Ee,\partial)$, $\abs{V}=1$, $\abs{\Ee}=2$ and the Laplacian defined on it by the boundary conditions
\begin{eqnarray*}
A= \begin{bmatrix} 1 & -1 \\ 0 & 0 \end{bmatrix} &\mbox{and} & B= \begin{bmatrix} 0 & 0 \\ 1 & 1 \end{bmatrix}.
\end{eqnarray*}
This operator $-\Delta_{aux}=-\Delta(A,B)$ has the scattering matrix $S_{aux}(k^2)= \begin{bmatrix}
0 & 1 \\ 1 & 0 \end{bmatrix}$. Now one glues together two external edges of $\Ge$ and replaces them by an internal edge of length $a$. The scattering matrix of this new problem is 
\begin{eqnarray*}
S(k^2)\ast_2 \begin{bmatrix} e^{ik a} & 0 \\ 0 & 1 \end{bmatrix} \begin{bmatrix} 0 & 1 \\ 1 & 0 \end{bmatrix}. 
\end{eqnarray*}    
\end{remark}

\begin{example}\label{exscattering}
Consider again Example~\ref{ex2IndefQG}. The graph described there can be obtained also by gluing two star graphs, where each is a copy of the graph given in Example~\ref{ex1IndQG}. These are glued together along the negative edges, which gives a new negative internal edge of a certain length $\au=\{a\}$. The coefficient matrices $\Cm^j=\Cm^j(k,ik)$ are
\begin{eqnarray*}
\Cm^j=\begin{bmatrix}
\chi_{++}^j(k,ik) & \chi_{+-}^j(k,ik) \\
\chi_{-+}^j(k,ik) & \chi_{--}^j(k,ik)
\end{bmatrix}
= \begin{bmatrix}i & (1-i) \\ (1+i) & -i\end{bmatrix}, & \mbox{for } j=1,2,
\end{eqnarray*}
see Example~\ref{exCm}. Using the gluing formula from Proposition~\ref{glue} with new edge length $a$ one obtains 
\begin{align*}
s_{11}(k)=&\chi^1_{++} + \chi^1_{+-} e^{-ka} \chi^2_{--} e^{-ka}[\mathds{1} - \chi^1_{--}e^{-ka}\chi^2_{--}e^{-ka}]^{-1}\chi^1_{-+} \\
=& i \tanh(ka) 
\end{align*}
and
\begin{align*}
s_{21}(k)=& \chi^2_{+-}e^{-ka}[\mathds{1}- \chi^1_{--}e^{-ka}\chi^2_{--}e^{-ka}]^{-1}\chi^1_{-+}\\
=& \frac{1}{\cosh(ka)}.
\end{align*}
This is well defined for any $k>0$. For symmetry reasons one has $s_{22}(k)=s_{11}(k)$ and $s_{21}(k)=s_{12}(k)$. Applying the gluing formula yields the same result as the direct calculation in Example~\ref{ex2IndefQG}. Note that the limits
\begin{eqnarray*}
\lim_{k\to 0} S(k^2) = \begin{bmatrix}0 & 1 \\ 1& 0 \end{bmatrix} &\mbox{and} & \lim_{k\to \infty} S(k^2) = \begin{bmatrix}i & 0 \\ 0 & i \end{bmatrix}
\end{eqnarray*}
exist. For small energies there is full transmission, whereas for large energies there is full reflection along with a phase shift. In addition, the scattering matrix resembles for large energies the direct sum of the scattering matrices of the two building stones, compare Example~\ref{S1}. This is a feature not known for scattering systems involving only semi--bounded operators. 
\end{example}

\subsection{Scattering and cloaking phenomena?}
In \cite{Bouchitte} the light propagation through metamaterials is described using regularizations of an indefinite operator 
\begin{align*}
\tau u= -\div A(\cdot)\grad u, &&  A(x)= \begin{cases}+1, & x\in\Omega_+, \\  -1, & x\in\Omega_-,     \end{cases}
\end{align*}
where
\begin{align*}
&\Omega_+= \{x\in \R^2 \mid \norm{x}\leq R_1\} \cup \{x\in \R^2 \mid \norm{x}\geq R_2\} \  \mbox{and} \\
 &\Omega_-=\{x\in \R^2 \mid R_1 < \norm{x} <R_2   \}
\end{align*}
with $R_2>R_1>0$. The quantity computed there for the measurement of the light propagation is the Dirichlet--to--Neumann--map for certain exterior domains. In the setting of the Sturm--Liouville theory the Dirichlet--to--Neumann--map is given by the Weyl--Titchmarsh $m$--function, which is also closely related to the scattering matrix. A generalization of the Weyl--Titchmarsh $m$--function to partial differential operators on exterior domains has been given recently in \cite{BehrndtRohleder}. As in the one dimensional case the generalization of the $m$--function is again the Dirichlet--to--Neumann--map. One can suspect that the generalization of the $m$--function is also closely related to the scattering matrix and to inverse scattering problems. 

The analysis of the Dirichlet--to--Neumann--map in \cite{Bouchitte} leads G.~Bouchitt{\'e} and B.~ Schweizer to the interpretation that in the considered situation a cloaking phenomenon is taking place. Cloaking simplified means that there is an object in the system that can be made invisible for an observer looking at the system from outside. Of course, one can make any object invisible by putting it into a kind of box and decoupling the observer from the object of interest inside. Cloaking means rather that the measurement of the observer is the same in both cases, when the object is present and when it is not. 

Since at least in dimension $d=1$ the Dirichlet--to--Neumann--map which is the measured quantity in \cite{Bouchitte} is related to scattering properties of the system it is reasonable to consider also scattering problems involving indefinite operators. The gluing formula given in Proposition~\ref{glue} allows the qualitative analysis of the scattering properties of a quasi--one--dimensional system involving compact components, on which the symbol of the operator is not sign--definite. In particular one can observe that even very small negative components have strong influence on the scattering matrix and yield to new features as the one discussed in Example~\ref{exscattering}. However the author has not been not able to identify cloaking phenomena in the system considered. For a further investigation one can think of studying the asymptotic behaviour of the scattering matrix for small and for large energies. It is not self--evident that these limits exist in all cases.

\section{Summary and outline}
An explicitly solvable model problem for second order differential operators with (discontinuous) sign--changing but invertible symbol has been introduced. All self--adjoint extensions have been parametrized using methods from extension theory. The crucial observation of this note is that there is a one--to--one correspondence between the parametrization of the self--adjoint realizations of the model operator $\tau$ which acts as ``plus Laplace'' or ``minus Laplace'' and the parametrization of the self--adjoint realizations of the positive Laplacian. This can be used to distinguish one realization among the set of realizations. This is done in two steps. First one considers the natural self--adjoint realization of the Laplace operator. This is defined by the so called standard or Kirchhoff boundary conditions and it can be parametrized using a unitary map in the deficiency spaces of the minimal Laplacian. Taking now the same unitary map to parametrize the extensions of the model operator $\tau$ gives a distinguished self--adjoint extensions of $T^{\min}$. It turns out that this realization is exactly the natural realization of $\tau$, which is also obtained using indefinite quadratic forms. Natural means that the domain of the operator is the natural domain for the composition $-\tfrac{d}{dx} A(\cdot)\tfrac{d}{dx}$. 

The conjecture is that this holds in more general situations and that this allows to show the self--adjointness of sign--indefinite operators of the type $-\div A(\cdot)\grad$ in their natural domain in $L^2(\Omega)$ with $A(x)$ elliptic for $x\in \Omega_+$ and $-A(x)$ elliptic for $x\in \Omega_-$, where $\Omega= \Omega_+ \cup \Omega_-$. In a very particular radially symmetric situation this has been verified. The scheme presented here at least gives a distinguished self--adjoint realization of such operators and for the model problem with finite deficiency indices the conjecture is indeed true. 

The spectrum of the self--adjoint realizations of the model operator resembles the spectrum of $-\Delta$ and $+\Delta$, where $-\Delta$ denotes the Laplacian on certain subgraphs. Using explicit formulae for the resolvents, the wave operators and the scattering matrix have been calculated. Even if there are only very small components where the symbol of the operator is not sign--definite this has strong influence on the qualitative behaviour of the scattering matrix. Particularly the scattering matrix exhibits new features and does not just resemble the scattering matrix of a direct sum of sign--definite operators. 

It constitutes an open problem whether the sign--indefinite operators that are considered in this paper arise naturally as limits of certain operators on domains or manifolds of dimension greater than one when the ``thickness'' shrinks to zero. For Laplacians on graphs there are such results, see \cite{PostExner05, Grieser} and the references therein. In the mentioned articles it has been shown that some self--adjoint Laplacians on finite metric graph are obtained as the limit of a Laplacian on a ``thick graph'' which is a thin neighbourhood in $\R^n$ of the original graph. For Laplace operators on twisted tubes similar results have been obtained recently, see \cite{David}. In particular it has been shown that the limit of such operators yields certain Schr\"{o}dinger operators on the real line.

\appendix
\section{}\label{AppA}

\begin{proof}[Proof of Proposition~\ref{resolvent}]

In order to prove that the kernel $r(\cdot,\cdot,k)$ defines the resolvent operator $(T(A,B)-k^2)^{-1}$ one has to check:
\begin{itemize}
\item[(i)] $(T^{\max}-k^2) \int_{\Ge} r(x,y,k)\varphi(y)=\varphi(x)$ for all $\varphi\in \He$, 
\item[(ii)] $R(k^2)\varphi=\int_{\Ge} r(x,y,k)\varphi(y)\in \Dom(T(A,B))$ for $\varphi\in\He$ and
\item[(iii)]  $r(x,y,-\overline{k})=r(y,x,k)^*$ holds for $\Im k>0$.
\end{itemize}
The statements (i) and (ii) prove that $R(k^2)$ is the left inverse of $T(A,B)-k^2$. To prove that it is also the right inverse it is sufficient to verify (iii). The proofs of (i) and (ii) are given for $k\in \Qe$ with (iii) the claim carries over to $k\in \Pe$.

Proof of (i): For $\varphi\in \He$ with $\varphi_j\in C_0^{\infty}(I_j,\C)$ for every $j\in \Ee\cup \Ie$, one uses that the Green's function to the problem is known and one has 
\begin{align*}
\left( -\frac{d^2}{dx^2} -k^2\right) \int_{I_j} \frac{i}{2k} e^{ik \abs{x_j-y_j}} \varphi_j(y_j)dy_j =&\varphi_j(x_j), && \mbox{for } k\in \Qe\cup \Pe, \\
-\left( +\frac{d^2}{dx^2} -k^2\right) \int_{I_j} \frac{1}{2k} e^{-k \abs{x_j-y_j}} \varphi_j(y_j)dy_j =&\varphi_j(x_j), && \mbox{for } k\in \Qe \ \mbox{and} \\
-\left( +\frac{d^2}{dx^2} -k^2\right) \int_{I_j} \frac{1}{2k} e^{k \abs{x_j-y_j}} \varphi_j(y_j)dy_j =&\varphi_j(x_j), && \mbox{for } k\in \Pe. 
\end{align*}
Since the kernel $r^0_{n,m}(x,y,k,i\kappa)$ defines a bounded operator in $\He$, by continuous continuation from 
\begin{align*}
\De^{\prime}= \bigoplus_{j\in \Ee \cup \Ie}  C_0^{\infty}(I_j, \C)
\end{align*}
to $\He$ it follows that for $k\in \Qe$ and for all $\varphi \in \He$ 
$$\left(T^{\max}-k^2\right) \int_{\Ge}  r^0(x,y,k,ik)\varphi(y)dy =\varphi(x)$$
holds. Observe furthermore that 
$$\left(T^{\max}-k^2\right) \int_{\Ge}  r^1(x,y,k,ik)\varphi(y)dy =0$$
which completes the proof.

Proof of (ii): Observe that 
\begin{eqnarray*}
\int_{\Ge} r^0(x,y,k,i\kappa)\varphi(y)\in \De &\mbox{and} & \int_{\Ge} r^1(x,y,k,i\kappa)\varphi(y)\in \De, 
\end{eqnarray*}
and hence $$R(k^2)\varphi=\int_{\Ge} r(\cdot,y,k)\varphi(y) \in \De=\Dom(T^{\max}).$$ Let be $k\in \Qe$ and consider the traces for the free Green's function $r^0(\cdot,\cdot,k,i\kappa)$. One has 
\begin{align*}
\underline{\int_{\Ge}r_0(x,y,k,i\kappa)f(y)}=& R_{n,m}^+(k,i\kappa,\au)^{-1}J_{n,m} \hat{f}(k,i\kappa), \\ \underline{\left(\int_{\Ge}r_0(x,y,k,i\kappa)f(y)\right)}^{\prime}=& I_{n,m}(k,i\kappa) R_{n,m}^+(k,i\kappa,\au)^{-1} J_{n,m}\hat{f}(k,i\kappa)
\end{align*}
with 
$$\hat{f}(k,i\kappa):= \int_{\Ge} \Phi_{n,m}(y,k,i\kappa) W_{n,m}(k,i\kappa) f(y) dy.$$
Since the free Green's function decouples the positive from the negative edges the above statement follows already from the corresponding calculation for self--adjoint Lapalce operators, compare \cite[Proof of Lemma 4.2]{VKRS2006}

For the correction term $r^1(\cdot,\cdot,k,i\kappa)$ observe that for an appropriate vector $f$ one has the traces
\begin{align*}
\underline{\begin{bmatrix} e^{ikx} & 0 & 0 \\ 0 & e^{ikx} & e^{-ikx} \end{bmatrix}\begin{bmatrix} f_{\Ee} \\ f_{\Ie_-} \\ f_{\Ie_+} \end{bmatrix}}= \begin{bmatrix} 1 & 0 & 0 \\ 0 & 1 & 1 \\ 0 & e^{ik\au} & e^{-ik\au} \end{bmatrix}\begin{bmatrix} f_{\Ee} \\ f_{\Ie_-} \\ f_{\Ie_+}  \end{bmatrix},
\end{align*}
\begin{align*}
\underline{\begin{bmatrix} e^{ikx} & 0 & 0 \\ 0 & e^{ikx} & e^{-ikx} \end{bmatrix}\begin{bmatrix} f_{\Ee} \\ f_{\Ie_-} \\ f_{\Ie_+} \end{bmatrix}}^{\prime}= \begin{bmatrix} ik & 0 & 0 \\ 0 & ik & -ik \\ 0 & -ik e^{ik\au} & ik e^{-ik\au} \end{bmatrix}\begin{bmatrix} f_{\Ee} \\ f_{\Ie_-} \\ f_{\Ie_+}  \end{bmatrix},
\end{align*}
and therefore
\begin{align*} 
\underline{\Phi_{n,m}(x,k,i\kappa)\hat{f}}= X_{n,m}(k,i\kappa)\hat{f}, && \underline{\left(\Phi_{n,m}(x,k,i\kappa)\hat{f}\right)}^{\prime}= Y_{n,m}(k,i\kappa)\hat{f},
\end{align*}
where $\hat{f}$ is here short for $\hat{f}(k,i\kappa)$, compare also \cite[Proof of Lemma 4.2]{VKRS2006}. This gives the traces of the correction term
\begin{align*}
\underline{\int_{\Ge}r_1(x,y,k,i\kappa)f(y)}= 
 X_{n,m}(k,i\kappa)G_{n,m}(k,i\kappa,\au) \hat{f}
\end{align*}
and
\begin{align*} 
 \underline{\left(\int_{\Ge}r_1(x,y,k,i\kappa)f(y)\right)}^{\prime}=
  Y_{n,m}(k,i\kappa)G_{n,m}(k,i\kappa,\au)  \hat{f}. 
\end{align*}
Together one has
\begin{align*}
&A \underline{\int_{\Ge}r^0(x,y,k,i\kappa)f(y)} + 
B  \underline{\left(\int_{\Ge}r^0(x,y,k,i\kappa)f(y)\right)}^{\prime}\\
+ &A \underline{\int_{\Ge}r^1(x,y,k,i\kappa)f(y)}
+ B\underline{\left(\int_{\Ge}r^1(x,y,k,i\kappa)f(y)\right)}^{\prime} \\ \\
= &A R_{n,m}^+(k,i\kappa,\au)^{-1}J_{n,m} \hat{f}+ B I_{n,m}(k,i\kappa) R_{n,m}^+(k,i\kappa,\au)^{-1} J_{n,m}\hat{f}\\
+ &A X_{n,m}(k,i\kappa)G_{n,m}(k,i\kappa,\au) \hat{f}+ B Y_{n,m}(k,i\kappa)G_{n,m}(k,i\kappa,\au)  \hat{f} \\ =&0,
\end{align*}
where one has used the definition of $G_{n,m}(k,i\kappa,\au)$ in Proposition~\ref{resolvent}. It follows that $$R(k^2)f = \int_{\Ge} r(\cdot,y,k)\varphi(y) \in \Dom(T(A,B)).$$

Proof of (iii): 
As $\frac{i}{2k}e^{ik \abs{x-y}}$ is the Green's function for the problem on $L^2(\R)$, the symmetry holds for $r^0(x,y,k,ik)$ with $k\in \Qe$. It remains to prove the symmetry of the correction term $r^1(x,y,k,ik)$. Substituting formula~\eqref{resonance} into the formula for the correction term for appropriate $k,\kappa$ gives
\begin{align*}
r^1(x,y,k,i\kappa)= 
  &\Phi_{n,m}(x,k,i\kappa) (R_{n,m}^+)^{-1}(k,i\kappa,\au) [\mathds{1} - \Cm(k,i\kappa) T_{n,m}(k,i\kappa) ]^{-1}\circ \\
\circ &\mathfrak{X}(k,i\kappa) (R_{n,m}^+)^{-1}(k, i\kappa) J_{n,m}\Phi_{n,m}(y,k,i\kappa)^T W_{n,m}(k, i\kappa).
\end{align*}
Notice that $W_{n,m}$ commutes with $R^+_{n,m}$ as well as with $\Phi_{n,m}$ and $J_{n,m}$. Consider now \\ $r(x,y,k,ik)$ for $k\in \Qe$. Using for the diagonal matrices $I_{n,m}$ and $J_{n,m}$ the equalities 
\begin{eqnarray*}
I_{n,m}^*J_{n,m}= I_{n,m} & \mbox{and} & I_{n,m}J_{n,m}= I_{n,m}^*
\end{eqnarray*}
one gets 
\begin{align*}
&\left\{[\mathds{1} - \Cm(k,ik) T_{n,m}(k,ik) ]^{-1} \Cm(k,ik)  W_{n,m}(k,ik) J_{n,m}\right\}^*  \\
= &\left\{[\mathds{1} - \Cm(k,ik) T_{n,m}(k,ik) ]^{-1} \Cm(k,ik)  W_{n,m}(k,ik) J_{n,m}\right\}^*  \\
=&  J_{n,m} W_{n,m}(k,ik)^*  \Cm(k,ik)^* [\mathds{1} -  T_{n,m}(k,ik)^* \Cm(k,ik)^* ]^{-1} \\
=& \frac{i}{-2\overline{k}}I_{n,m} J_{n,m}\Cm(k,ik)^* [\mathds{1} -  T_{n,m}(k,ik)^* \Cm(k,ik)^* ]^{-1} \\
=& \frac{i}{-2\overline{k}} 
[\mathds{1} -  (I_{n,m}^* \Cm(k,ik)^*) I_{n,m} (T_{n,m}(k,ik)^* ) ]^{-1} (I_{n,m}^*\Cm(k,ik)^*) 
\end{align*}
Using the formulas given in Lemma~\ref{cmsym} one continues as follows
\begin{align*}
&\frac{i}{-2\overline{k}} 
[\mathds{1} -  (I_{n,m}^* \Cm(k,ik)^*) I_{n,m} (T_{n,m}(k,ik)^* ) ]^{-1} (I_{n,m}^*\Cm(k,ik)^*) \\
=&\frac{i}{-2\overline{k}}[\mathds{1} -  \Cm(-\overline{k},i\overline{k}) T_{n,m}(k,ik)^*   ]^{-1} \Cm(-\overline{k},i\overline{k}) I_{n,m}^* \\
=&[\mathds{1} -  \Cm(-\overline{k},i\overline{k}) T_{n,m}(-\overline{k},i\overline{k})   ]^{-1} \Cm(-\overline{k}, i\overline{k}) W_{n,m}(-\overline{k},i\overline{k})J_{n,m}\\
=&[\mathds{1} -  \Cm(-\overline{k}, i\overline{k}) T_{n,m}(-\overline{k},i\overline{k})   ]^{-1} \Cm(-\overline{k}, i\overline{k}) W_{n,m}(-\overline{k},i\overline{k})J_{n,m}.
\end{align*}
Furthermore one has 
\begin{eqnarray*}
R_{n,m}^+(k,ik,\au)^*= R_{n,m}^+(-\overline{k},i\overline{k},\au) & \mbox{and} & \Phi(x,k,ik)^*= \Phi(x,-\overline{k},i\overline{k}), 
\end{eqnarray*}
compare also \cite[Proof of Lemma 4.2]{VKRS2006}. Putting the pieces together one obtains
\begin{eqnarray*}
r^1(k,ik,y,x)^* = r^1(-\overline{k},i\overline{k},x,y) & \mbox{for } k\in  \Qe 
\end{eqnarray*}
which proves the claim.
\end{proof}

\section{}\label{AppB}
\begin{proof}[Continuation of the proof of Proposition \ref{wo}]

The computation of the wave operators is supplemented. One computes for $k^2=\lambda>0$
\begin{align*}
\lim_{\epsilon \to 0}-i\epsilon R(k^2\pm i\epsilon) dE_0(k)f=  \lim_{\epsilon \to 0} -i\epsilon\int_{\Ge}r(x,y,\sqrt{k^2\pm i\epsilon}) \begin{bmatrix} \cos(k x)\hat{f}_+(k) \\ 0  \end{bmatrix}   dy
\end{align*}
with $\{\cos(k x)\hat{f}_+(k)\}_{j\in \Ee_+}= \cos(k x_j)\hat{f}_j(k)_{j\in \Ee_+}$ and with  $0$ is meant the zero on the rest of the components. For $-\kappa^2=\lambda<0$ one computes
\begin{align*}
\lim_{\epsilon \to 0}-i\epsilon R(-\kappa^2\pm i\epsilon) dE_0(\kappa)f=  \lim_{\epsilon \to 0} -i\epsilon\int_{\Ge}r(x,y,\sqrt{-\kappa^2\pm i\epsilon}) \begin{bmatrix} 0 \\ \cos(\kappa x)\hat{f}_-(\kappa)   \end{bmatrix}  dy
\end{align*}
with $\{\cos(\kappa x)\hat{f}_-(\kappa)\}_{j\in \Ee_-}= \cos(\kappa x_j)\hat{f}_j(\kappa)_{j\in \Ee_-}$ and again $0$ denotes the zero on the rest of the components. 

Fix first the branch of the complex square root 
\begin{eqnarray*}
\sqrt{\cdot}\colon \C \setminus [0,\infty) \rightarrow \C^+ & \mbox{with } \Im \sqrt{\cdot}>0,
\end{eqnarray*}
where $\C^+=\{z\in \C \mid \Im z >0  \}$. Consider the limit values  
\begin{eqnarray*}
k_{\epsilon}^+= \sqrt{k^2 + i\epsilon}, &   &  k_{\epsilon}^-= \sqrt{k^2 - i\epsilon}, \\
\kappa_{\epsilon}^+= \sqrt{-\kappa^2 + i\epsilon}& \mbox{and} &  \kappa_{\epsilon}^-= \sqrt{-\kappa^2 - i\epsilon}.
\end{eqnarray*}
Taking the limit $\epsilon \to 0+$ for $k>0$ and $\kappa >0$, respectively gives
\begin{eqnarray*}
\lim_{\epsilon \to 0+} k_{\epsilon}^+ = k, &   & \lim_{\epsilon \to 0+}  k_{\epsilon}^-= -k, \\
\lim_{\epsilon \to 0+} \kappa_{\epsilon}^+ = i\kappa &  \mbox{and} & \lim_{\epsilon \to 0+}  \kappa_{\epsilon}^-= i\kappa.
\end{eqnarray*}
Observe that for $\epsilon \to 0+$ with fixed $k>0$ and $\kappa>0$ one has asymptotically
\begin{eqnarray*}
\lim_{\epsilon \to 0+} \left(\kappa_{\epsilon}^+ - i\kappa\right) \sim \frac{\epsilon}{2\kappa}, 
& & \lim_{\epsilon \to 0+} \left(\kappa_{\epsilon}^- - i\kappa\right) \sim \frac{\epsilon}{2\kappa},\\
\lim_{\epsilon \to 0+} \left(k_{\epsilon}^+ - k\right) \sim \frac{i\epsilon}{2k}
& \mbox{and} & \lim_{\epsilon \to 0+} \left(k_{\epsilon}^- + k\right) \sim \frac{i\epsilon}{2k}.
\end{eqnarray*}

The following auxiliary calculations are needed,
\begin{align*}
&-i \epsilon \int_{0}^{\infty} e^{ik_{\epsilon}^+ \abs{x-y}} \cos(ky) dy= \\ 
&\frac{-i\epsilon}{2} \int_{0}^{x} e^{ik_{\epsilon}^+ (x-y) + ik y } +  \frac{-i\epsilon}{2} \int_{x}^{\infty} e^{ik_{\epsilon}^+ (y-x) + ik y }  dy \\ 
+ &\frac{-i\epsilon}{2} \int_{0}^{x} e^{ik_{\epsilon}^+ (x-y) - ik y } + \frac{-i\epsilon}{2} \int_{x}^{\infty} e^{ik_{\epsilon}^+ (y-x) - ik y }  dy
= \\
&\frac{-i\epsilon}{2i(k-k_{\epsilon}^+)}  e^{ik_{\epsilon}^+ x} \left(e^{i(k-k_{\epsilon}^+)x } -1 \right) 
+  \frac{-i\epsilon}{2i (k_{\epsilon}^+ +k)} e^{-ik_{\epsilon}^+ x}  \left(-e^{i(k_{\epsilon}^+ +k) x}  \right) \\ 
+&\frac{-i\epsilon}{-2i(k_{\epsilon}^++k)}  e^{ik_{\epsilon}^+ x} \left(e^{-i(k_{\epsilon}^++k)x } -1 \right) 
+\frac{-i\epsilon}{2i(k_{\epsilon}^+-k)}  e^{-ik_{\epsilon}^+ x} \left(-e^{i(k_{\epsilon}^+-k)x }\right).
\end{align*}
In the limit $\epsilon \to 0+$ all summands of the right hand side except the last one vanish. The last term becomes
$-ik e^{-ik x}$. Consequently
\begin{align*}
\lim_{\epsilon\to 0+}-i \epsilon \int_{0}^{\infty} e^{ik_{\epsilon}^+ \abs{x-y}} \cos(ky) dy= -ik e^{-ik x}.
\end{align*}
Analogously one obtains
\begin{align*}
&-i \epsilon \int_{0}^{\infty} e^{ik_{\epsilon}^+ y} \cos(ky) dy= \\
&\frac{-i\epsilon}{2} \int_{0}^{\infty} e^{ik_{\epsilon}^+ y + ik y } +  \frac{-i\epsilon}{2} \int_{0}^{\infty} e^{ik_{\epsilon}^+ y - ik y }  dy = \\
  &\frac{-i\epsilon}{2i (k_{\epsilon}^+ +k)}   \left(-1 \right) +\frac{-i\epsilon}{2i (k_{\epsilon}^+ -k)}   \left(-1 \right). 
\end{align*}
In the limit $\epsilon \to 0+$ only the second addend of the right hand side remains and takes the value $-ik$. Consequently
\begin{align*}
\lim_{\epsilon\to 0+} -i \epsilon \int_{0}^{\infty} e^{ik_{\epsilon}^+ y} \cos(ky) dy= -ik.
\end{align*}
Now the same calculation with $k_{\epsilon}^-$ yields
\begin{align*}
&-i \epsilon \int_{0}^{\infty} e^{ik_{\epsilon}^- \abs{x-y}} \cos(ky) dy= \\ 
&\frac{-i\epsilon}{2} \int_{0}^{x} e^{ik_{\epsilon}^- (x-y) + ik y } +  \frac{-i\epsilon}{2} \int_{x}^{\infty} e^{ik_{\epsilon}^- (y-x) + ik y }  dy \\ 
+ &\frac{-i\epsilon}{2} \int_{0}^{x} e^{ik_{\epsilon}^- (x-y) - ik y } + \frac{-i\epsilon}{2} \int_{x}^{\infty} e^{ik_{\epsilon}^- (y-x) - ik y }  dy
= \\
&\frac{-i\epsilon}{2i(k-k_{\epsilon}^-)}  e^{ik_{\epsilon}^- x} \left(e^{i(k-k_{\epsilon}^-)x } -1 \right) 
+  \frac{-i\epsilon}{2i (k_{\epsilon}^- +k)} e^{-ik_{\epsilon}^- x}  \left(-e^{i(k_{\epsilon}^- +k) x}  \right) \\ 
+&\frac{-i\epsilon}{-2i(k_{\epsilon}^-+k)}  e^{ik_{\epsilon}^- x} \left(e^{-i(k_{\epsilon}^-+k)x } -1 \right) 
+\frac{-i\epsilon}{2i(k_{\epsilon}^--k)}  e^{-ik_{\epsilon}^- x} \left(-e^{i(k_{\epsilon}^- -k)x }\right).
\end{align*}
In the limit all terms on the right hand side except the second term $$\frac{-i\epsilon}{2i (k_{\epsilon}^- +k)} e^{-ik_{\epsilon}^- x}  \left(-e^{i(k_{\epsilon}^- +k) x}  \right)$$ vanish, which becomes
$-ik e^{ik x}$. Consequently 
\begin{align*}
\lim_{\epsilon\to 0+} -i \epsilon \int_{0}^{\infty} e^{ik_{\epsilon}^- \abs{x-y}} \cos(ky) dy= -ik e^{ik x}.
\end{align*}

Analogously one obtains
\begin{align*}
-i \epsilon \int_{0}^{\infty} e^{ik_{\epsilon}^- y} \cos(ky) dy&= 
\frac{-i\epsilon}{2} \int_{0}^{\infty} e^{ik_{\epsilon}^- y + ik y } +  \frac{-i\epsilon}{2} \int_{0}^{\infty} e^{ik_{\epsilon}^- y - ik y }  dy  \\
  &=\frac{-i\epsilon}{2i (k_{\epsilon}^- +k)}   \left(-1 \right) +\frac{-i\epsilon}{2i (k_{\epsilon}^- -k)}   \left(-1 \right). 
\end{align*}
In the limit  $\epsilon \to 0+$ only the second term remains and takes the value $-ik$. Consequently 
\begin{align*}
\lim_{\epsilon\to 0+} -i \epsilon \int_{0}^{\infty} e^{ik_{\epsilon}^- y} \cos(ky)= -ik.
\end{align*}

Taking into account the above auxiliary calculations yields
\begin{align*}
&\lim_{\epsilon \to 0+} -i \epsilon \int_{\Ge}r(x,y,\sqrt{k^2+ i\epsilon}) \begin{bmatrix} \cos(k x)\hat{f}(k) \\ 0  \end{bmatrix}   dy \\
=&\lim_{\epsilon \to 0+} -i \epsilon\int_{\Ge}r^0(x,y,\sqrt{k^2+ i\epsilon},i \sqrt{k^2+ i\epsilon}) \begin{bmatrix} \cos(k x)\hat{f}(k) \\ 0  \end{bmatrix}   dy \\
 +&\lim_{\epsilon \to 0+} -i \epsilon\int_{\Ge}r^1(x,y,\sqrt{k^2+ i\epsilon}, i\sqrt{k^2+ i\epsilon}) \begin{bmatrix} \cos(k x)\hat{f}(k) \\ 0  \end{bmatrix}   dy \\
= &\frac{1}{2}\begin{bmatrix} e^{-ikx}\hat{f}(k) \\ 0  \end{bmatrix}+
\frac{1}{2} \Phi_{n,m}(x,k,ik)G_{n,m}(k,ik) \begin{bmatrix} \hat{f}(k) \\ 0  \end{bmatrix} \\
=&\frac{1}{2} \sum_{l\in\Ee_+} \varphi_l(x,k,ik)\hat{f}_l(k). 
\end{align*}
Analogue
\begin{align*}
&\lim_{\epsilon \to 0+} \int_{\Ge}r(x,y,\sqrt{k^2- i\epsilon}) \begin{bmatrix} \cos(k x)\hat{f}(k) \\ 0  \end{bmatrix}   dy \\
= &\lim_{\epsilon \to 0+} \int_{\Ge}r^0(x,y,\sqrt{k^2- i\epsilon},-i\sqrt{k^2- i\epsilon}) \begin{bmatrix} \cos(k x)\hat{f}(k) \\ 0  \end{bmatrix}   dy \\
 + &\lim_{\epsilon \to 0+} \int_{\Ge}r^1(x,y,\sqrt{k^2- i\epsilon},-i\sqrt{k^2- i\epsilon}) \begin{bmatrix} \cos(k x)\hat{f}(k) \\ 0  \end{bmatrix}   dy \\
= &\frac{1}{2}\begin{bmatrix} e^{ikx}\hat{f}(k) \\ 0  \end{bmatrix}+
\frac{1}{2} \Phi_{n,m}(x,-k,ik)G_{n,m}(-k,ik) \begin{bmatrix} \hat{f}(k) \\ 0  \end{bmatrix}\\
=&\frac{1}{2} \sum_{l\in\Ee_+} \varphi_l(x,-k,ik)\hat{f}_l(k). 
\end{align*}

Using similar calculation gives
\begin{align*}
\lim_{\epsilon \to 0+} -i \epsilon\int_{\Ge}r(x,y,\sqrt{-\kappa^2+ i\epsilon}) \begin{bmatrix} 0 \\ \cos(\kappa x)\hat{f}(\kappa)  \end{bmatrix}   dy 
=\frac{1}{2} \sum_{l\in\Ee_-} \varphi_l(x,ik,-k)\hat{f}_l(k)
\end{align*}
and 
\begin{align*}
\lim_{\epsilon \to 0+} -i \epsilon\int_{\Ge}r(x,y,\sqrt{-\kappa^2- i\epsilon}) \begin{bmatrix} 0 \\ \cos(\kappa x)\hat{f}(\kappa)  \end{bmatrix}   dy 
=\frac{1}{2} \sum_{l\in\Ee_-} \varphi_l(x,ik,k)\hat{f}_l(k).
\end{align*}

Using 
\begin{equation*}
\displaystyle{ W_{\pm}f = \int_{-\infty}^{\infty} \lim_{\varepsilon \to 0+} - i \varepsilon  R(\lambda\mp i\varepsilon) dE_0(\lambda)f  ,}
\end{equation*}
one obtains that the kernel of the wave operators is given in terms of the generalized eigenfunctions.
\end{proof}




\end{document}